\documentclass[twoside,leqno,twocolumn]{article}  
\usepackage{ltexpprt}

\usepackage{amsmath,amssymb,graphicx,epsfig,color}
\usepackage{enumerate}

\newcommand{\eps}{\ensuremath{\varepsilon}}
\newcommand{\poly}[1]{\text{poly}\left(#1\right) }
\newcommand{\ignore}[1]{}

\newcommand{\RR}{\mathbb{R}}
\newcommand{\NN}{\mathbb{N}}

\newcommand{\Prob}[1]{\ensuremath{\mathbb{P}\left(#1\right)}}
\DeclareMathOperator{\EE}{\mathbb{E}}
\newcommand{\pinv}[1]{ {#1}^\dagger}
\newcommand{\norm}[1]{\ensuremath{\left\|#1\right\|_2}}

\newcommand{\frobnorm}[1]{\ensuremath{\left\|#1\right\|_{\text{\rm F}}}}
\newcommand{\sr}[1]{\ensuremath{\mathrm{\textbf{\footnotesize sr}}\left(#1\right)}}
\newcommand{\trace}[1]{\ensuremath{\mathrm{\textbf{tr}}\left(#1\right)}}
\newcommand{\rank}[1]{\ensuremath{\mathrm{\textbf{{\footnotesize rank}}}\left(#1\right)}}
\newcommand{\im}[1]{\ensuremath{\mathrm{\textbf{Im}}\left(#1\right)}}

\usepackage[backref=page]{hyperref}
\hypersetup{
    unicode=false,          % non-Latin characters in Acrobat’s bookmarks
    pdftoolbar=true,        % show Acrobat’s toolbar?
    pdfmenubar=true,        % show Acrobat’s menu?
    pdffitwindow=true,      % page fit to window when opened
    pdftitle={Low Rank Matrix-valued Chernoff Bounds and Approximate Matrix Multiplication},    %
    pdfauthor={Avner Magen, Anastasios Zouzias},     % author
    pdfsubject={Approximation algorithms for matrix multiplication},   % subject of the document
    pdfcreator={Anastasios Zouzias},   % creator of the document
    pdfproducer={Anastasios Zouzias, Avner Magen}, % producer of the document
    pdfkeywords={Chernoff Bounds, matrix-valued,Khinthine Inequality, least squares, regression matrix multiplication, Singular Value Decomposition}, % list of keywords
    pdfnewwindow=true,      % links in new window
    colorlinks=false,       % false: boxed links; true: colored links
    linkcolor=red,          % color of internal links
    citecolor=green,        % color of links to bibliography
    filecolor=magenta,      % color of file links
    urlcolor=cyan          % color of external links
}

\renewcommand*{\backref}[1]{}
\renewcommand*{\backrefalt}[4]{%
	\ifcase #1 %
		(Not cited) %
	\or
		(Cited on page~#2)%
	\else 
		(Cited on pages~#2)%
	\fi
}

\newtheorem{remark}{Remark}
\newtheorem{claim}{Claim}

\begin{document}

\title{\Large Low Rank Matrix-valued Chernoff Bounds and Approximate Matrix Multiplication
%\thanks{(TODO)}
}
\author{Avner Magen
\thanks{University of Toronto, Department of Computer Science, Email : \url{avner@cs.toronto.edu}. } \\
\and 
Anastasios Zouzias
\thanks{University of Toronto, Department of Computer Science, Email : \url{zouzias@cs.toronto.edu}.}
}
\date{}

\maketitle

%\pagenumbering{arabic}
%\setcounter{page}{1}%Leave this line commented out.

\begin{abstract} 
\small\baselineskip=9pt In this paper we develop algorithms for approximating matrix multiplication with respect to the spectral norm. Let $A\in{\RR^{n\times m}}$ and $B\in{\RR^{n \times p}}$ be two matrices and $\eps>0$. We approximate the product $A^\top B$ using two sketches $\widetilde{A}\in{\RR^{t\times m}}$ and $\widetilde{B}\in{\RR^{t\times p}}$, where $t\ll n$, such that
\begin{equation*}
 \norm{\widetilde{A}^\top \widetilde{B} - A^\top B} \leq \eps \norm{A}\norm{B}
\end{equation*}
with high probability. We analyze two different sampling procedures for constructing $\widetilde{A}$ and $\widetilde{B}$; one of them is done by i.i.d. non-uniform sampling rows from $A$ and $B$ and the other by taking random linear combinations of their rows. We prove bounds on $t$ that depend only on the intrinsic dimensionality of $A$ and $B$, that is their rank and their stable rank. \par For achieving bounds that depend on rank when taking random linear combinations we employ standard tools from high-dimensional geometry such as concentration of measure arguments combined with elaborate $\eps$-net constructions. For bounds that depend on the smaller parameter of stable rank this technology itself seems weak. However, we show that in combination with a simple truncation argument it is amenable to provide such bounds. To handle similar bounds for row sampling, we develop a novel matrix-valued Chernoff bound inequality which we call low rank matrix-valued Chernoff bound. Thanks to this inequality, we are able to give bounds that depend only on the stable rank of the input matrices. \par We highlight the usefulness of our approximate matrix multiplication bounds by supplying two applications. First we give an approximation algorithm for the $\ell_2$-regression problem that returns an approximate solution by randomly projecting the initial problem to dimensions linear on the rank of the constraint matrix. Second we give improved approximation algorithms for the low rank matrix approximation problem with respect to the spectral norm.
\end{abstract}

%%%%%%%%%%%%%%%%%%%%%%%%%%%%%%%%%%%%%%%%%%%%%%%%%%%%%%%%%%%%%%%%%%%
%%%%%%%%%%%%%%%%%%%%%%%%%%%%%%%%%%%%%%%%%%%%%%%%%%%%%%%%%%%%%%%%%%%
\section{Introduction}\label{sec:tools}
%%%%%%%%%%%%%%%%%%%%%%%%%%%%%%%%%%%%%%%%%%%%%%%%%%%%%%%%%%%%%%%%%%%
%%%%%%%%%%%%%%%%%%%%%%%%%%%%%%%%%%%%%%%%%%%%%%%%%%%%%%%%%%%%%%%%%%%

In many scientific applications, data is often naturally expressed as a matrix, and computational problems on such data are reduced to standard matrix operations including matrix multiplication, $\ell_2$-regression, and low rank matrix approximation. 

In this paper we analyze several approximation algorithms with respect to these operations. All of our algorithms share a common underlying framework which can be described as follows: Let $A$ be an input matrix that we may want to apply a matrix computation on it to infer some useful information about the data that it represents. The main idea is to work with a sample of $A$ (a.k.a. sketch), call it $\widetilde{A}$, and hope that the obtained information from $\widetilde{A}$ will be in some sense close to the information that would have been extracted from $A$.

In this generality, the above approach (sometimes called ``Monte-Carlo method for linear algebraic problems'') is ubiquitous, and is responsible for much of the development in fast matrix computations~\cite{lowrank:FKV,matrixmult:drineas,sarlos,l2_regression:drineas06,matrix:sparsification:optas,CW_stoc09,matrix:volume_sampling:FOCS2010}.

As we sample $A$ to create a sketch $\widetilde{A}$, our goal is twofold: (\textit{i}) guarantee that $\widetilde{A}$ resembles $A$ in the relevant measure, and (\textit{ii}) achieve such a $\widetilde{A}$ using as few samples as possible. The standard tool that provides a handle on these requirements when the objects are real numbers, is the Chernoff bound inequality. However, since we deal with matrices, we would like to have an analogous probabilistic tool suitable for matrices. Quite recently a non-trivial generalization of Chernoff bound type inequalities for matrix-valued random variables was introduced by Ahlswede and Winter~\cite{chernoff:matrix_valued:AW}. Such inequalities are suitable for the type of problems that we will consider here. However, this type of inequalities and their variants that have been proposed in the literature \cite{chernoff:matrix_valued:Bernstein:Gross,recht:simple_completion,chernoff:matrix_valued:Gross,chernoff:matrix_valued:Tropp} all suffer from the fact that their bounds depend on the dimensionality of the samples. We argue that in a wide range of applications, this dependency can be quite detrimental. 

Specifically, whenever the following two conditions hold we typically provide stronger bounds compared with the existing tools: (\textit{a}) the input matrix has low intrinsic dimensionality such as rank or stable rank, (\textit{b}) the matrix samples themselves have low rank. The validity of condition (\textit{a}) is very common in applications from the simple fact that viewing data using matrices typically leads to redundant representations. Typical sampling methods tend to rely on extremely simple sampling matrices, i.e., samples that are supported on only one entry~\cite{matrix:sparsification:arora,matrix:sparsification:optas,matrix:sparsification:zouzias} or samples that are obtained by the outer-product of the sampled rows or columns~\cite{matrixmult:drineas,lowrank:rankone:VR}, therefore condition (\textit{b}) is often natural to assume. By incorporating the rank assumption of the matrix samples on the above matrix-valued inequalities we are able to develop a ``dimension-free'' matrix-valued Chernoff bound. See Theorem~\ref{thm:chernoff:matrix_valued:low_rank} for more details. 

%%%%%%%%%%%%%%%%%%%%%%%%%%%%%%%%%%%%%%%%%%%%%%%%%%%%%%%%%%%%%%%%%%%
%%%%%%%%%%%%%%%%%%%%%%%%%%%%%%%%%%%%%%%%%%%%%%%%%%%%%%%%%%%%%%%%%%%
Fundamental to the applications we derive, are two probabilistic tools that provide concentration 
bounds of certain random matrices. These tools are inherently different, where each pertains to a 
different sampling procedure. In the 
first, we multiply the input matrix by a random sign matrix, whereas in the second we sample
rows according to a distribution that depends on the input matrix.
In particular,  the first 
method is oblivious (the probability space does not depend on the input matrix)
 while the second is not.

The first tool is the so-called subspace Johnson-Lindenstrauss lemma. Such a result was obtained in \cite{sarlos} (see also~\cite[Theorem~1.3]{jl:manifold}) although it appears implicitly in results extending the original Johnson Lindenstrauss lemma (see~\cite{magen07}). The techniques for proving such a result with possible worse bound are not new and can be traced back even to Milman's proof of Dvoretsky theorem~\cite{Dvoretsky:Milman}.
\begin{lemma}\label{lem:jl_subspace} (Subspace JL lemma \cite{sarlos})
Let $\mathcal{W} \subseteq \RR^d$ be a linear subspace of dimension $k$ and
$\eps\in{(0, 1/3)}$. Let $R$ be a $t\times d$ random sign matrix rescaled by $1/\sqrt{t}$, namely $R_{ij} = \pm 1/\sqrt{t}$ with equal probability.
Then
\begin{eqnarray} 
\Prob{ (1-\eps) \norm{w}^2 \leq \norm{Rw}^2 \leq (1+\eps)\norm{w}^2,\ \forall\ w\in\mathcal{W} } \nonumber \\
 \geq 1 - c_2^k \cdot \exp (- c_1 \eps^2 t),\label{eq:jl_subspace}
\end{eqnarray}
where $c_1>0,c_2>1$ are constants.
\end{lemma}
The importance of such a tool, is that it allows us to get bounds on the necessary dimensions of the random sign matrix in terms of the \emph{rank} of the input matrices, see Theorem~\ref{thm:matrixmult} (\textit{i.a}).

While the assumption that the input matrices have low rank is a fairly reasonable assumption, one should be a little cautious as the property of 
having low rank is not robust. Indeed, if random noise is added to a matrix, even if low rank, the matrix obtained will have full rank almost 
surely. On the other hand, it can be shown that the added noise cannot distort the Frobenius and operator norm significantly; which makes the notion of {\em stable rank} robust and so the assumption of low stable rank on the input is more applicable than the low rank assumption.

Given the above discussion, we resort to a different methodology, called matrix-valued Chernoff bounds. These are non-trivial generalizations 
of the standard Chernoff bounds over the reals and were first introduced in~\cite{chernoff:matrix_valued:AW}.  Part of the contribution of the current work is to show that such inequalities, similarly to their real-valued ancestors, provide powerful tools to analyze randomized algorithms. There is a rapidly growing line of research exploiting the power of such inequalities including matrix approximation by sparsification~\cite{matrix:sparsification:optas,matrix:sparsification:zouzias}; analysis of algorithms for matrix completion and decomposition of low rank matrices~\cite{chernoff:matrix_valued:Candes_Sparse,chernoff:matrix_valued:Gross,recht:simple_completion}; and semi-definite relaxation and rounding of quadratic maximization problems~\cite{chernoff:matrix_valued:opt:Nemirovski,chernoff:matrix_valued:opt,chernoff:matrix_valued:opt:journal}.

The quality of these bounds can be measured by the number of samples needed in order to obtain small error probability. The original result of 
\cite[Theorem~19]{chernoff:matrix_valued:AW} shows that\footnote{For ease of presentation we actually provide the restatement presented in~\cite[Theorem~2.6]{chernoff:matrix_valued:derand:WX08}, which is more suitable for this discussion.} if $M$ is distributed according to some distribution over $n \times n$ matrices with zero mean\footnote{Zero mean means that the (matrix-valued) expectation is the zero $n\times n$ matrix.}, and if $M_1,\dots ,M_t$ are independent copies of $M$ then for any $\eps>0$, 
\begin{equation}\label{ineq:chernoff:naive}
\Prob{\norm{\frac1{t}\sum_{i=1}^{t} M_i} > \eps} \leq n \exp\left(- C 
\frac{\eps^2 t}{\gamma^2}\right),
\end{equation}
where $\norm{M}\leq \gamma$ holds almost surely and $C>0$ is an absolute constant.

Notice that the number of samples in Ineq.~\eqref{ineq:chernoff:naive} depends logarithmically in $n$. In general, unfortunately, such a dependency is inevitable: take for example a diagonal random sign matrix of dimension $n$. The operator norm of the sum of $t$ independent samples is precisely the 
maximum deviation among $n$ independent random walks of length $t$. In order to achieve a fixed bound on the maximum deviation with constant probability, it is easy to see that $t$ should grow logarithmically with $n$ in this scenario.

In their seminal paper, Rudelson and Vershynin provide a matrix-valued Chernoff bound that avoids the dependency on the dimensions by assuming that the matrix samples are the \emph{outer product} $x\otimes x$ of a randomly distributed vector $x$~\cite{lowrank:rankone:VR}. It turns out that this assumption is too strong in most applications, such as the ones we study in this work, and so we wish to relax it without increasing the bound significantly. In the following theorem we replace this assumption with that of having \emph{low rank}. We should note that we are not aware of a simple way to extend Theorem~$3.1$ of~\cite{lowrank:rankone:VR} to the low rank case, even constant rank. The main technical obstacle is the use of the powerful Rudelson selection lemma, see~\cite{rudelson:isotropic} or Lemma~$3.5$ of~\cite{lowrank:rankone:VR}, which applies only for Rademacher sums of outer product of vectors. We bypass this obstacle by proving a more general lemma, see Lemma~\ref{lem:E_p_vs_sum_of_squares}. The proof of Lemma~\ref{lem:E_p_vs_sum_of_squares} relies on the non-commutative Khintchine moment inequality~\cite{khintchine:LP86,khintchine:Buchholz} which is also the backbone in the proof of Rudelson's selection lemma. With Lemma~\ref{lem:E_p_vs_sum_of_squares} at our disposal, the proof techniques of~\cite{lowrank:rankone:VR} can be adapted to support our more general condition.
\begin{theorem}\label{thm:chernoff:matrix_valued:low_rank}
Let $0<\eps <1$ and $M$ be a random symmetric real matrix with $\norm{\EE{M}}\leq 1$ and $\norm{M} \leq \gamma$ almost surely. Assume that each
element on the support of $M$ has at most rank $r$. Set $t=\Omega(\gamma \log (\gamma/\eps^2)  /\eps^2)$. If $r\leq t$ holds almost surely, then 
\begin{equation*}
 \Prob{ \norm{ \dfrac1{t}\sum_{i=1}^{t}{M_i} - \EE M } >\eps }~\leq~ \dfrac1{\poly{t} }.
\end{equation*}
where $M_1,M_2, \dots , M_t$ are i.i.d. copies of $M$.
%\end{enumerate}
\end{theorem}
\begin{proof}
See Appendix, page~\pageref{sec:chernoff:matrix_valued:low_rank}. 
\end{proof}
\begin{remark}[Optimality]
The above theorem cannot be improved in terms of the number of samples required without changing its form, since in the special case where the rank of the samples is one it is exactly the statement of~Theorem~$3.1$ of \cite{lowrank:rankone:VR}, see~\cite[Remark~$3.4$]{lowrank:rankone:VR}.
\end{remark}
\begin{table*}[ht]
{\small
%\hfill{}
\centering
    \begin{tabular}{ | c | c | c | c | c |}
    \hline
    \multicolumn{5}{|c|}{Variants of Matrix-valued Inequalities} \\
    \hline
    \emph{Assumption on the sample} $M$     & \# \emph{of samples} ($t$)    & Failure Prob. & \emph{References}      & \emph{Comments} 
\\ \hline\hline
    $\norm{M}\leq \gamma $ a.s.	    & $\Omega (\gamma^2  \log (n) /\eps^{2})$ & $1/\poly{n}$ & \cite{chernoff:matrix_valued:derand:WX08}   & {\small Hoeffding} 
\\ \hline
    $\norm{M}\leq \gamma $ a.s., $\norm{ \EE M^2} \leq \rho^2$    & $\Omega ( ( \rho^2 + \gamma \eps/3 )  \log (n) /\eps^{2})$ & $1/\poly{n}$ & \cite{recht:simple_completion}   & {\small Bernstein} 
\\ \hline
 $\norm{M} \leq \gamma $ a.s., $M=x\otimes x$, $\norm{\EE{M}}\leq 1$ & $\Omega(\gamma \log (\gamma
/\eps^2)/ \eps^2 )$ & $\exp (-\Omega(\epsilon^2 t/(\gamma \log t) ))$ & \cite{lowrank:rankone:VR} & {\small Rank one} 
\\ \hline
 $\norm{M} \leq \gamma $, $\rank{M} \leq t$ a.s., $\norm{\EE{M}}\leq 1$ & $\Omega( \gamma  \log (\gamma /\eps^2) /\eps^2)$ & $1/\poly{t}$ & {\small Theorem~\ref{thm:chernoff:matrix_valued:low_rank}} & {\small Low rank}
\\
\hline
\end{tabular}}
%\hfill{}
\caption{Summary of matrix-valued Chernoff bounds. $M$ is a probability
distribution over symmetric $n\times n$ matrices. $M_1,\dots ,M_t$ are i.i.d. copies of $M$.}
\end{table*}
We highlight the usefulness of the above main tools by first proving a ``dimension-free'' approximation algorithm for matrix multiplication with respect to the spectral norm (Section~\ref{sec:apps:matrix_mult}). Utilizing this matrix multiplication bound we get an approximation algorithm for the $\ell_2$-regression problem which returns an approximate solution by randomly projecting the initial problem to dimensions linear on the rank of the constraint matrix (Section~\ref{sec:apps:l2_regression}). Finally, in Section~\ref{sec:apps:low_rank} we give improved approximation algorithms for the low rank matrix approximation problem with respect to the spectral norm, and moreover answer in the affirmative a question left open by the authors of~\cite{low_rank:STOC09}.
%%%%%%%%%%%%%%%%%%%%%%%%%%%%%%%%%%%%%%%%%%%%%%%%%%%%%%%%%%%%%%%%%%%
\section{Preliminaries and Definitions}
%%%%%%%%%%%%%%%%%%%%%%%%%%%%%%%%%%%%%%%%%%%%%%%%%%%%%%%%%%%%%%%%%%%
The next discussion reviews several definitions and facts from linear algebra; for more details, see~\cite{book:perturbation:stewart,book:GVL,book:matrix:Bhatia}. We abbreviate the terms independently and identically distributed and almost surely with i.i.d. and a.s., respectively. We let $\mathbb{S}^{n-1}:=\{x\in\RR^n~|~\norm{x}=1\}$ be the $(n-1)$-dimensional sphere. A \emph{random Gaussian} matrix is a matrix whose entries are i.i.d. standard Gaussians, and a \emph{random sign} matrix is a matrix whose entries are independent Bernoulli random variables, that is they take values from $\{\pm 1\}$ with equal probability. For a matrix $A\in\RR^{n\times m}$, $A_{(i)}$, $A^{(j)}$, denote the $i$'th row, $j$'th column, respectively. For a matrix with rank $r$, the Singular Value Decomposition (SVD) of $A$ is the decomposition of $A$ as $U\Sigma V^\top$ where $U\in{\RR^{n\times r}}$, $V\in{\RR^{m\times r}}$ where the columns of $U$ and $V$ are orthonormal, and $\Sigma = \text{diag}(\sigma_1(A),\dots , \sigma_r(A))$ 
is $r\times r$ diagonal matrix. We further assume $\sigma_1 \geq \ldots \geq \sigma_r > 0$ and call these real numbers the {\em singular values} of $A$. By $A_k=U_k \Sigma_k V_k^\top$ we denote the best rank $k$ approximation to $A$, where $U_k$ and $V_k$ are the matrices formed by the
first $k$ columns of $U$ and $V$, respectively. We denote by $\norm{A}=\max \{ \norm{Ax}~|~\norm{x} =1 \}$ the spectral norm of $A$, and by
$\frobnorm{A}=\sqrt{\sum_{i,j}{A_{ij}^2}}$ the Frobenius norm of $A$. We denote by $\pinv{A}$ the Moore-Penrose pseudo-inverse of $A$, i.e.,
$\pinv{A}=V\Sigma^{-1} U^\top$. Notice that $\sigma_1(A)=\norm{A}$. Also we define by $\sr{A}:=\frobnorm{A}^2/\norm{A}^2$ the \emph{stable rank} of $A$. Notice that the inequality $\sr{A} \leq \rank{A}$ always holds. The orthogonal projector of a matrix $A$ onto the row-space of a matrix $C$ is denoted by $P_C (A) = A\pinv{C}C$. By $P_{C,k}(A)$ we define the best rank-$k$ approximation of the matrix $P_C (A)$.

%%%%%%%%%%%%%%%%%%%%%%%%%%%%%%%%%%%%%%%%%%%%%%%%%%%%%%%%%%%%%%%%%%%
%%%%%%%%%%%%%%%%%%%%%%%%%%%%%%%%%%%%%%%%%%%%%%%%%%%%%%%%%%%%%%%%%%%
\section{Applications}\label{sec:apps}
%%%%%%%%%%%%%%%%%%%%%%%%%%%%%%%%%%%%%%%%%%%%%%%%%%%%%%%%%%%%%%%%%%%
%%%%%%%%%%%%%%%%%%%%%%%%%%%%%%%%%%%%%%%%%%%%%%%%%%%%%%%%%%%%%%%%%%%
All the proofs of this section have been deferred to Section~\ref{sec:proofs}.
%%%%%%%%%%%%%%%%%%%%%%%%%%%%%%%%%%%%%%%%%%%%%%%%%%%%%%%%%%%%%%%%%%%
%%%%%%%%%%%%%%%%%%%%%%%%%%%%%%%%%%%%%%%%%%%%%%%%%%%%%%%%%%%%%%%%%%%
\subsection{Matrix Multiplication}\label{sec:apps:matrix_mult}
%%%%%%%%%%%%%%%%%%%%%%%%%%%%%%%%%%%%%%%%%%%%%%%%%%%%%%%%%%%%%%%%%%%
%%%%%%%%%%%%%%%%%%%%%%%%%%%%%%%%%%%%%%%%%%%%%%%%%%%%%%%%%%%%%%%%%%%
The seminal research of~\cite{lowrank:FKV} focuses on using non-uniform row sampling to speed-up the running 
time of several matrix computations. The subsequent developments of~\cite{matrixmult:drineas,lowrank:drineas, matrixdecomp:drineas} 
also study the performance of Monte-Carlo algorithms on primitive matrix algorithms including the matrix multiplication problem with 
respect to the Frobenius norm. Sarlos~\cite{sarlos} extended (and improved) this line of research using random projections. Most of the 
bounds for approximating matrix multiplication in the literature are mostly with respect to the Frobenius 
norm~\cite{matrixmult:drineas, sarlos, CW_stoc09}. In some cases, the 
techniques that are utilized for bounding the Frobenius norm also imply \emph{weak} bounds for the spectral norm,
 see~\cite[Theorem~4]{matrixmult:drineas} or~\cite[Corollary~11]{sarlos} which is similar with part (\textit{i.a}) of Theorem~\ref{thm:matrixmult}. 

In this section we develop approximation algorithms for matrix
multiplication with respect to the spectral norm. The algorithms that will be
presented in this section are based on the tools mentioned in Section~\ref{sec:tools}. Before stating our main dimension-free matrix multiplication theorem 
(Theorem~\ref{thm:matrixmult}), we discuss the best possible bound 
that can be achieved using the current known matrix-valued inequalities (to the best of our knowledge).
Consider a direct application of Ineq.~\eqref{ineq:chernoff:naive},
where a similar analysis with that in proof of Theorem~\ref{thm:matrixmult} (\textit{ii}) would allow us to achieve a bound of $\Omega(\widetilde{r}^2 \log (m+p) /\eps^2)$ on the number of samples (details omitted). 
However, as the next theorem indicates (proof omitted) we can get linear dependency on 
the stable rank of the input matrices gaining from the ``variance information'' of the samples; more precisely, this can be achieved by applying the matrix-valued Bernstein Inequality~see e.g.~\cite{chernoff:matrix_valued:Bernstein:Gross}, \cite[Theorem~3.2]{recht:simple_completion} or~\cite[Theorem~2.10]{chernoff:matrix_valued:Tropp}.
\begin{theorem}
Let $0< \eps < 1/2$ and let $A\in{\RR^{n\times m}}$, $B\in{\RR^{ n\times p}}$
both having stable rank at most $\widetilde{r}$. The
following hold:
\begin{enumerate}[(i)]
 \item 
Let $R$ be a $t\times n$ random sign matrix rescaled by $1/\sqrt{t}$. Denote by $\widetilde{A}=RA$ and $\widetilde{B}=RB$. If $t=\Omega(\widetilde{r} \log (m+p)/\eps^2 )$ then
\[ \Prob{\norm{\widetilde{A}^\top \widetilde{B} - A^\top B} \leq \eps \norm{A} \norm{B}
} \geq 1- \frac1{\poly{\widetilde{r}}  }. \]
\item
Let $p_i = \norm{A_{(i)}}\norm{B_{(i)}} /S$, where $S=\sum_{i=1}^{n}{\norm{A_{(i)}}\norm{B_{(i)}}}$ be a probability distribution over $[n]$. If we form a $t\times m$ matrix $\widetilde{A}$ and a $t\times p$ matrix $\widetilde{B}$ by
taking $t=\Omega(\widetilde{r} \log (m + p) /\eps^2)$ i.i.d. (row indices) samples from $p_i$, then
\[ \Prob{\norm{\widetilde{A}^\top \widetilde{B} - A^\top B} \leq \eps \norm{A} \norm{B}
} \geq 1- \frac1{\poly{\widetilde{r}}}. \]
\end{enumerate}
\end{theorem}
 Notice that the above bounds depend linearly on the stable rank of the matrices and logarithmically on their dimensions. 
As we will see in the next theorem we can remove the dependency on the dimensions, and replace it with the stable rank. Recall that in most cases matrices \emph{do} have low stable rank, which is much smaller that their 
dimensionality.
%\clearpage
%%%%%%%%%%%%%%%%%%%%%%%%%%%%%%%%%%%%%%%%%%
%		Matrix Multiplication
%%%%%%%%%%%%%%%%%%%%%%%%%%%%%%%%%%%%%%%%%%
%%%%%%%%%%%%%%%%%%%%%%%%%%%%%%%%%%%%%%%%%%
\begin{theorem}\label{thm:matrixmult}
Let $0< \eps < 1/2$ and let $A\in{\RR^{n\times m}}$, $B\in{\RR^{ n\times p}}$
both having rank and stable rank at most $r$ and $\widetilde{r}$, respectively. The
following hold:
\begin{enumerate}[(i)]
 \item 
Let $R$ be a $t\times n$ random sign matrix rescaled by $1/\sqrt{t}$. Denote by $\widetilde{A}=RA$ and $\widetilde{B}=RB$.
\begin{enumerate}[(a)]
 \item 
 If $t=\Omega(r/\eps^{2} )$ then
\[ \mathbb{P}( \forall x\in\RR^m, y\in\RR^p, \  |x^\top (\widetilde{A}^\top \widetilde{B} - A^\top B)y|\]
\[ \leq \eps \norm{Ax} \norm{By}) \geq 1- e^{-\Omega(r)}.\]
 \item
If $t=\Omega(\widetilde{r}/\eps^4 )$ then
\[ \Prob{\norm{\widetilde{A}^\top \widetilde{B} - A^\top B} \leq \eps \norm{A} \norm{B}
} \geq 1- e^{-\Omega( \frac{\widetilde{r}}{\eps^2} ) }. \]
\end{enumerate}
\item
Let $p_i = \norm{A_{(i)}}\norm{B_{(i)}} /S$, where $S=\sum_{i=1}^{n}{\norm{A_{(i)}}\norm{B_{(i)}}}$ be a probability distribution over $[n]$. If we form a $t\times m$ matrix $\widetilde{A}$ and a $t\times p$ matrix $\widetilde{B}$ by
taking $t=\Omega(\widetilde{r}  \log ( \widetilde{r}/\eps^2)  /\eps^2)$ i.i.d. (row indices) samples from $p_i$, then
\[ \Prob{\norm{\widetilde{A}^\top \widetilde{B} - A^\top B} \leq \eps \norm{A} \norm{B}
} \geq 1- \frac1{\poly{\widetilde{r}}}. \]
\end{enumerate}
\end{theorem}
\begin{remark}
In part (\textit{ii}), we can actually achieve the \emph{stronger} bound of $t=\Omega(\sqrt{\sr{A}\sr{B}}\log ( \sr{A}$ $\sr{B} /\eps^4)  /\eps^2)$ (see proof). However, for ease of presentation and comparison we give the above displayed bound.
\end{remark}
Part (\textit{i.b}) follows from (\textit{i.a}) via a simple truncation argument. This was pointed out to us by Mark Rudelson~(personal communication). To understand the significance and the differences between the different components of this theorem, we first note that the probabilistic event of part (\textit{i.a}) is superior to the probabilistic event of (\textit{i.b}) and (\textit{ii}). Indeed, when $B=A$ the former implies that $|x^\top (\widetilde{A}^\top \widetilde{A} - A^\top A) x| < \eps \cdot x^\top A^\top A x$ for every $x$, which is stronger than $\norm{\widetilde{A}^\top \widetilde{A} - A^\top A} \leq \eps \norm{A}^2$. We will \emph{heavily} exploit this fact in Section~\ref{sec:app:spectral} to prove Theorem~\ref{thm:lowrank} (\textit{i.a}) and (\textit{ii}). Also notice that part (\textit{i.b}) is essential computationally inferior to (\textit{ii}) as it gives the same bound while it is more expensive computationally to multiply the matrices by random sign matrices than just sampling their rows. However, the advantage of part (\textit{i}) is that the sampling process is \emph{oblivious}, i.e., does not depend on the input matrices.
\ignore{
% TASOS : Maybe we should place this sentence somewhere.
We also note that the special case of part (\textit{ii}) where $A=B$ is precisely ~\cite[Theorem~3.1]{lowrank:rankone:VR}. 
In its present generality this theorem is tight as can be seen by the reduction of~\cite[Theorem~2.8]{CW_stoc09} 
\footnote{This reduction deals with the Frobenius norm and so applicable here as always $\norm{\cdot} \leq  \frobnorm{\cdot}$}. 
 However, we don't know if this bound holds in the special case of $B=A$. 
In a nutshell, the importance of deriving tights bounds for approximate matrix multiplication lies on the fact that in 
many linear algebraic problems are, after manipulations, reduced to primitive problems including matrix multiplication.
}
%%%%%%%%%%%%%%%%%%%%%%%%%%%%%%%%%%%%%%%%%%
%%%%%%%%%%%%%%%%%%%%%%%%%%%%%%%%%%%%%%%%%%
%%%%%%%%%%%%%%%%%%%%%%%%%%%%%%%%%%%%%%%%%%
%%%%%%%%%%%%%%%%%%%%%%%%%%%%%%%%%%%%%%%%%%
%		l_2 regression
%%%%%%%%%%%%%%%%%%%%%%%%%%%%%%%%%%%%%%%%%%
%%%%%%%%%%%%%%%%%%%%%%%%%%%%%%%%%%%%%%%%%%
\subsection{$\ell_2$-regression}\label{sec:apps:l2_regression}
%%%%%%%%%%%%%%%%%%%%%%%%%%%%%%%%%%%%%%%%%%
%%%%%%%%%%%%%%%%%%%%%%%%%%%%%%%%%%%%%%%%%%
In this section we present an approximation algorithm for the least-squares
regression problem; given an $n\times m$, $n>m$, real matrix $A$ of rank $r$ and a real
vector $b\in\RR^n$ we want to compute $x_{\text{opt}}=\pinv{A}b$ that minimizes
$\norm{Ax-b}$ over all $x\in\RR^m$. In their seminal paper~\cite{l2_regression:drineas06}, Drineas et al. show that if we
non-uniformly sample $t=\Omega(m^2/\eps^2)$ rows from $A$ and $b$, then with high probability the optimum 
solution of the $t\times d$ sampled problem will be within $(1+\eps)$ close to the original problem. 
The main drawback of their approach is that finding or even approximating
the sampling probabilities is computationally intractable.
Sarlos~\cite{sarlos} improved the above to 
$t=\Omega( m\log m/\eps^2)$ and gave the first $o(nm^2)$ 
relative error approximation algorithm for this problem. 

In the next theorem we eliminate the extra $\log m$ factor from Sarlos bounds, 
and more importantly, replace the dimension (number of variables) $m$ with the 
rank $r$ of the constraints matrix $A$.
We should point out that independently, the same bound as our 
Theorem~\ref{thm:ell2_regression} was recently obtained by Clarkson and
Woodruff~\cite{CW_stoc09} (see also~\cite{l2_regression:drineas:faster}). 
The proof of Clarkson and Woodruff uses heavy machinery and a completely 
different approach. In a nutshell they manage to improve the matrix multiplication bound with respect to the Frobenius norm. They achieve
this by bounding higher moments of the Frobenius norm of the approximation 
viewed as a random variable instead of bounding the \emph{local} differences for
each coordinate of the product. To do so, they rely on intricate moment
calculations spanning over four pages, see~\cite{CW_stoc09} for more. On the other hand, the proof of the
present $\ell_2$-regression bound uses only basic matrix analysis, elementary
deviation bounds and $\eps$-net arguments. More precisely, we argue that Theorem~\ref{thm:matrixmult} (\textit{i.a}) immediately 
implies that by randomly-projecting to dimensions linear in the intrinsic dimensionality of the constraints, i.e., the rank of $A$, is
sufficient as the following theorem indicates.
%%%%%%%%%%%%%%%%%%%%%%%%%%%%%%%%%%%%%%%%%%
%		l_2 regression
%%%%%%%%%%%%%%%%%%%%%%%%%%%%%%%%%%%%%%%%%%
%%%%%%%%%%%%%%%%%%%%%%%%%%%%%%%%%%%%%%%%%%
\begin{theorem}\label{thm:ell2_regression}
Let $A\in{\RR^{n\times m}}$ be a real matrix of rank $r$ and $b\in\RR^n$. Let $\min_{x\in\RR^m} \norm{b-Ax}$ be the $\ell_2$-regression problem, where the minimum is achieved with $x_{opt}=\pinv{A}b$. Let $0<\eps<1/3$, $R$ be a $t\times n$ random sign matrix rescaled by $1/\sqrt{t}$ and $\widetilde{x}_{opt}=\pinv{(RA)} Rb$. 
\begin{itemize}
 \item 
If $t=\Omega(r/\eps)$, then with high probability,
\begin{equation}\label{ineq:regression:approx}
 \norm{b-A\widetilde{x}_{opt}} \leq (1+\eps) \norm{b-Ax_{opt}}.
\end{equation}
\item
If $t=\Omega(r/\eps^2)$, then with high probability,
\begin{equation}\label{ineq:regression:x_opt}
 \norm{x_{opt} - \widetilde{x}_{opt}} \leq
\dfrac{\eps}{\sigma_{\min}(A)}\norm{b-Ax_{opt}}.
\end{equation}
\end{itemize}
\end{theorem}
\begin{remark}
The above result can be easily generalized to the case where $b$ is an $n\times p$ matrix $B$ of rank at most $r$ (see proof). This is known as the generalized $\ell_2$-regression problem in the literature, i.e., $\arg\min_{X\in{m\times p}}\norm{AX-B}$ where $B$ is an $n\times p$ rank $r$ matrix.
\end{remark}
%%%%%%%%%%%%%%%%%%%%%%%%%%%%%%%%%%%%%%%%%%%%%%%
%%%%%%%%%%%%%%%%%%%%%%%%%%%%%%%%%%%%%%%%%%%%%%%
\subsection{Spectral Low Rank Matrix Approximation}\label{sec:apps:low_rank}
%%%%%%%%%%%%%%%%%%%%%%%%%%%%%%%%%%%%%%%%%%%%%%%
%%%%%%%%%%%%%%%%%%%%%%%%%%%%%%%%%%%%%%%%%%%%%%%
A large body of work on low rank matrix approximations~\cite{lra:drineaskannan,lowrank:FKV,lowrank:volume_sampling:DRVW06,sarlos,lowrank:rankone:VR,matrix:sparsification:optas,low_rank:Tygert2008,CW_stoc09,low_rank:STOC09,matrix:survey:HMT09} 
has been recently developed with main objective to develop more efficient algorithms
for this task. Most of these results study approximation algorithms with respect to the Frobenius norm, except for
\cite{lowrank:rankone:VR, low_rank:STOC09} that handle the spectral norm.

In this section we present two $(1+\eps)$-relative-error approximation algorithms for this problem with
 respect to the spectral norm, i.e., given an $n\times m$, $n> m$, real matrix $A$ of rank $r$, 
 we wish to compute $A_k=U_k \Sigma_k V^\top_k$, which minimizes $\norm{A-X_k}$ over the set of
$n\times m$ matrices of rank $k$, $X_k$. The first additive bound for this problem was obtained in~\cite{lowrank:rankone:VR}. 
To the best of our knowledge the best relative bound was recently achieved in~\cite[Theorem~1]{low_rank:STOC09}. 
The latter result is not directly comparable with ours, since it uses a more restricted projection methodology and 
so their bound is weaker compared to our results. The first algorithm randomly projects the rows of the 
input matrix onto $t$ dimension. Here, we set $t$ to be either $\Omega(r/\eps^2)$ in which case we get
an $(1+\eps)$ error guarantee, or to be $\Omega(k/\eps^2)$ in which case we show a $(2+\eps\sqrt{(r-k)/k})$ error approximation. In both cases the algorithm succeeds with high probability. 
The second approximation algorithm samples non-uniformly $\Omega(r \log (r/\eps^2) /\eps^2)$ rows from $A$ 
in order to satisfy the $(1+\eps)$ guarantee with high probability.

The following lemma (Lemma~\ref{lem:rayleigh_implies_lowrank}) is essential for
proving both relative error bounds of Theorem~\ref{thm:lowrank}. It gives a sufficient condition
that any matrix $\widetilde{A}$ should satisfy in order to get a $(1+\eps)$
spectral low rank matrix approximation of $A$ for \emph{every} $k$, $1\leq k \leq
\rank{A}$. 
\begin{lemma}\label{lem:rayleigh_implies_lowrank}
Let $A$ be an $n \times m$ matrix and $\eps>0$. If there exists a $t\times m$
matrix $\widetilde{A}$ such that for every $x\in{\RR^m}$, $(1-\eps)x^\top A^\top A x
\leq x^\top \widetilde{A}^\top \widetilde{A} x \leq (1+\eps) x^\top A^\top A x$, then
\begin{equation*}
 \norm{A- P_{\widetilde{A},k}(A)} \leq (1+\eps) \norm{A - A_k},
\end{equation*}
for \emph{every} $k=1,\dots, \rank{A}$.
\end{lemma}
The theorem below shows that it's possible to satisfy the conditions of
Lemma~\ref{lem:rayleigh_implies_lowrank} by randomly projecting $A$ onto
$\Omega(r/\eps^2)$ or by non-uniform sampling i.i.d. $\Omega(r \log(r/\eps^2) /\eps^2)$
rows of $A$ as described in parts (\textit{i.a}) and (\textit{ii}), respectively.
%%%%%%%%%%%%%%%%%%%%%%%%%%%%%%%%%%%%%%%%%%
%		Low Rank Approximation
%%%%%%%%%%%%%%%%%%%%%%%%%%%%%%%%%%%%%%%%%%
%%%%%%%%%%%%%%%%%%%%%%%%%%%%%%%%%%%%%%%%%%
\begin{theorem}\label{thm:lowrank}
Let $0<\eps <1/3$ and let  $A=U\Sigma V^\top$ be a real $n \times m$ matrix of
rank $r$ with $n\geq m$.
\begin{enumerate}[(i)]
\item 
\begin{enumerate}[(a)]
\item
Let $R$ be a $t\times n$ random sign matrix rescaled by $1/\sqrt{t}$ and set $\widetilde{A}=RA$. If $t=\Omega(r/\eps^2)$, then with high probability
\[	\norm{A-P_{\widetilde{A},k}(A)} \leq (1+\eps) \norm{A-A_k},\]
for \emph{every} $k=1,\dots,r$.
\item
Let $R$ be a $t\times n$ random Gaussian matrix rescaled by $1/\sqrt{t}$ and set $\widetilde{A}=RA$. If $t=\Omega(k/\eps^2)$, then with high probability
\[	\norm{A-P_{\widetilde{A},k}(A)} \leq (2+\eps \sqrt{\frac{r-k}{k}}) \norm{A-A_k}.\]
\end{enumerate}
\item
Let $p_i=\norm{U_{(i)}}^2 /r$ be a probability distribution over $[n]$. Let
$\widetilde{A}$ be a $t\times m$ matrix that is formed (row-by-row) by taking $t$
i.i.d. samples from $p_i$ and rescaled appropriately. If $t=\Omega(r 
\log (r/\eps^2) /\eps^{2})$, then with high probability
\[	\norm{A-P_{\widetilde{A},k}(A) } \leq (1+\eps) \norm{A-A_k},\]
for \emph{every} $k=1,\dots,r$.
\end{enumerate}
\end{theorem}
We should highlight that in part (\textit{ii}) the probability distribution $p_i$ is in
general hard to compute. Indeed, computing $\norm{U_{(i)}}^2$ requires computing the SVD of $A$. In general, these values are known as statistical leverage scores~\cite{matrix:leverage_scores:drineas}. In the special case where $A$ is an edge-vertex matrix of an
undirected weighted graph then $p_i$, the probability distribution over edges
(rows), corresponds to the effective-resistance of the $i$-th
edge~\cite{graph:sparsifiers:eff_resistance}. 
%%%%%%%%%%%%%%%%%%%%%%%%%%%%%%%%%%%%%%%%%%

%%%%%%%%%%%%%%%%%%%%%%%%%%%%%%%%%%%%%%%%%%
Theorem~\ref{thm:lowrank} gives an $(1+\eps)$ approximation algorithm for the special case of low rank matrices. However, as discussed in Section~\ref{sec:tools} such an assumption is too restrictive for most applications. In the following theorem, we make a step further and relax the rank condition with a condition that depends on the stable rank of the residual matrix $A-A_k$. More formally, for an integer $k\geq 1$, we say that a matrix $A$ has a \emph{$k$-low stable rank tail} iff $k \geq \sr{A-A_k}$. 

Notice that the above definition is useful since it contains the set of matrices whose spectrum follows a power-law distribution and those with exponentially decaying spectrum. Therefore the following theorem combined with the remark below (partially) answers in the affirmative the question posed by~\cite{low_rank:STOC09}: Is there a relative error approximation algorithm with respect to the spectral norm when the spectrum of the input matrix decays in a power law?
\begin{theorem}\label{thm:lowrank:low_stable_tail}
Let $0<\eps <1/3$ and let $A$ be a real $n \times m$ matrix with a \emph{$k$-low stable rank tail}. Let $R$ be a $t\times n$ random sign matrix rescaled by $1/\sqrt{t}$ and set $\widetilde{A}=RA$. If $t = \Omega(k / \eps ^4)$, then with high probability
\[	\norm{A-P_{\widetilde{A},k}(A)} \leq (2+\eps) \norm{A-A_k}.\]
\end{theorem}
\begin{remark}
The $(2+\eps)$ bound can be improved to a relative $(1+\eps)$ error bound if we return as the approximate solution a slightly higher rank matrix, i.e., by returning the matrix $P_{\widetilde{A}}(A)$, which has rank at most $t=\Omega(k/\eps^4)$ (see \cite[Theorem~$9.1$]{matrix:survey:HMT09}).
\end{remark}
%\clearpage
%%%%%%%%%%%%%%%%%%%%%%%%%%%%%%%%%%%%%%%%%%%%%%%%%%%%%%%%%%%%%%%%%%%
%%%%%%%%%%%%%%%%%%%%%%%%%%%%%%%%%%%%%%%%%%%%%%%%%%%%%%%%%%%%%%%%%%%
\section{Proofs}\label{sec:proofs}
%%%%%%%%%%%%%%%%%%%%%%%%%%%%%%%%%%%%%%%%%%%%%%%%%%%%%%%%%%%%%%%%%%%
%%%%%%%%%%%%%%%%%%%%%%%%%%%%%%%%%%%%%%%%%%%%%%%%%%%%%%%%%%%%%%%%%%%
%%%%%%%%%%%%%%%%%%%%%%%%%%%%%%%%%%%%%%%%%%%%%%%%%%%%%%%%%%%%%%%%%%%
%%%%%%%%%%%%%%%%%%%%%%%%%%%%%%%%%%%%%%%%%%%%%%%%%%%%%%%%%%%%%%%%%%% 
\subsection{Proof of Theorem~\ref{thm:matrixmult} (Matrix Multiplication)}
%%%%%%%%%%%%%%%%%%%%%%%%%%%%%%%%%%%%%%%%%%%%%%%%%%%%%%%%%%%%%%%%%%%
%%%%%%%%%%%%%%%%%%%%%%%%%%%%%%%%%%%%%%%%%%%%%%%%%%%%%%%%%%%%%%%%%%%
\paragraph{Random Projections - Part (\textit{i})}
%%%%%%%%%%%%%%%%%%%%%%%%%%%%%%%%%%%%%%%%%%%%%%%%%%%%%%%%%%%%%%%%%%%
%%%%%%%%%%%%%%%%%%%%%%%%%%%%%%%%%%%%%%%%%%%%%%%%%%%%%%%%%%%%%%%%%%%
\paragraph{Part (\textit{a}):} 
%%%%%%%%%%%%%%%%%%%%%%%%%%%%%%%%%%%%%%%%%%%%%%%%%%%%%%%%%%%%%%%%%%%
%%%%%%%%%%%%%%%%%%%%%%%%%%%%%%%%%%%%%%%%%%%%%%%%%%%%%%%%%%%%%%%%%%%
In this section we show the first, to the best of our knowledge, non-trivial
spectral bound for matrix multiplication. Although the proof is an immediate
corollary of the subspace Johnson-Lindenstrauss lemma
(Lemma~\ref{lem:jl_subspace}), this result is powerful enough to give, for example, tight
bounds for the $\ell_2$ regression problem. We prove the following more general theorem from which
Theorem~\ref{thm:matrixmult} (\textit{i.a}) follows by plugging in $t=\Omega(r/\eps^2)$.
\begin{theorem}\label{thm:matrixmult:restated}
Let $A\in{\RR^{n\times m}}$ and $B\in{\RR^{ n\times p}}$. Assume that the ranks
of $A$ and $B$ are at most $r$. Let $R$ be a $t\times n$ random sign matrix rescaled by $1/\sqrt{t}$. Denote by $\widetilde{A}= RA$ and $\widetilde{B}= RB$. The following inequality
holds
\[ \Prob{ \forall x\in\RR^m, y\in\RR^p, \quad  |x^\top (\widetilde{A}^\top \widetilde{B} - A^\top B)y| \leq \eps \norm{Ax} \norm{By} }\]
\[ \geq 1- c_2^{r} \exp (-c_1 \eps^2 t), \]
where $c_1>0,c_2>1$ are constants.
\end{theorem}
\begin{proof}(of Theorem~\ref{thm:matrixmult:restated})
Let $A=U_A\Sigma_A V_A^\top$, $B=U_B \Sigma_B V^\top_{B}$ be the singular value
decomposition of $A$ and $B$ respectively. Notice that $U_{A}\in{\RR^{n\times
r_A}},U_{B}\in{\RR^{n\times r_B}}$, where $r_A$ and $r_B$ is the rank of $A$ and
$B$, respectively.

Let $x_1\in{\RR^m},x_2\in{\RR^{p}}$ two arbitrary unit vectors. Let $w_1= A x_1$
and $w_2=B x_2$. Recall that 
\[\norm{A^\top R^\top RB - A^\top B} =\]
\[ \sup_{x_1\in{\mathbb{S}^{m-1}},x_2\in{\mathbb{S}^{p-1}} } | x_1^\top(A^\top
R^\top RB - A^\top B)x_2|.\] 
We will bound the last term for any arbitrary vector. Denote with $\mathcal{V}$ the subspace\footnote{We denote by $\text{colspan}(A)$ the subspace generated by the columns of $A$, and $\text{rowspan}(A)$ the subspace generated by the rows of $A$.} $\text{colspan}(U_A)\cup \text{colspan}(U_B)$ of $\RR^n$. Notice that the size of $dim(\mathcal{V}) \leq r_A + r_B \leq 2r$. Applying Lemma~\ref{lem:jl_subspace} to $\mathcal{V}$, we get that with probability at least $1-c_2^{r}\exp(-c_1\eps^2 t)$ that
\begin{equation}\label{eq:matrixmult}
\forall\ v \in{\mathcal{V}}: \  \ |\norm{Rv}^2- \norm{v}^2 | \leq \eps \norm{v}^2.
\end{equation}
Therefore we get that for any unit vectors $v_1,v_2\in{\mathcal{V}}$:
\begin{eqnarray*}
 (Rv_1)^\top Rv_2	&   =  & \dfrac{\norm{ Rv_1 + Rv_2}^2-\norm{Rv_1 - Rv_2}^2}{4}\\
			& \leq & \dfrac{(1+\eps)\norm{v_1+v_2}^2-(1-\eps)\norm{v_1-v_2}^2}{4}\\
  		    	&   =  & \dfrac{\norm{v_1+v_2}^2-\norm{v_1-v_2}^2}{4}\\
			&   +  & \eps \dfrac{\norm{v_1+v_2}^2+\norm{v_1-v_2}^2}{4}\\
			&   =  & v_1^\top v_2 + \eps \frac{\norm{v_1}^2+\norm{v_2}^2}{2}\ =\ v_1^\top v_2 + \eps,
\end{eqnarray*}
where the first equality follows from the Parallelogram law, the first inequality follows from Equation~\eqref{eq:matrixmult}, and the last inequality
since $v_1,v_2$ are unit vectors. By similar considerations we get that $(Rv_1)^\top Rv_2  \geq v_1^\top v_2 - \eps$. By linearity of $R$, we get that 
\[\forall v_1,v_2 \in{\mathcal{V} }: \  \ |(Rv_1)^\top Rv_2 - v_1^\top v_2 | \leq \eps \norm{v_1}\norm{v_2} .  \]
Notice that $w_1,w_2\in{\mathcal{V} }$, hence $ |w_1^\top R^\top R w_2 - w_1^\top w_2| \leq \eps \norm{w_1}\norm{w_2} = \eps \norm{Ax_1}\norm{Bx_2}$.
\end{proof}
%%%%%%%%%%%%%%%%%%%%%%%%%%%%%%%%%%%%%%%%%%%%%%%%%%%%%%%%%%%%%%%%%%%
%%%%%%%%%%%%%%%%%%%%%%%%%%%%%%%%%%%%%%%%%%%%%%%%%%%%%%%%%%%%%%%%%%%
\paragraph{Part (\textit{b}):} 
%%%%%%%%%%%%%%%%%%%%%%%%%%%%%%%%%%%%%%%%%%%%%%%%%%%%%%%%%%%%%%%%%%%
%%%%%%%%%%%%%%%%%%%%%%%%%%%%%%%%%%%%%%%%%%%%%%%%%%%%%%%%%%%%%%%%%%%
We start with a technical lemma that bounds the spectral norm of any matrix $A$ when it's multiplied by a random sign matrix rescaled by $1/\sqrt{t}$.
\begin{lemma}\label{lem:Rudelson}
Let $A$ be an $n\times m$ real matrix, and let $R$ be a $t\times n$ random sign matrix rescaled by $1/\sqrt{t}$. If $t\geq \sr{A}$, then
\begin{equation}
 \Prob{ \norm{ RA } \geq 4 \norm{A} }\ \leq\ 2e^{-t/2}.
\end{equation}
\end{lemma}
\begin{proof}
Without loss of generality assume that $\norm{A} = 1$. Then $\frobnorm{A} = \sqrt{\sr{A}}$. Let $G$ be a $t\times n$ Gaussian matrix. Then by the Gordon-Chev\`{e}t inequality\footnote{For example, set $S=I_t, T=A$ in \cite[Proposition~$10.1$,~p.~$54$]{matrix:survey:HMT09}.}
\begin{eqnarray*}
 	\EE{\norm{GA} }	& \leq & \norm{I_t}\frobnorm{A} + \frobnorm{I_t}\norm{A} \\
			&   =  & \frobnorm{A} + \sqrt{t}\ \leq\ 2\sqrt{t}.
\end{eqnarray*}
The Gaussian distribution is symmetric, so $G_{ij}$ and $\sqrt{t} R_{ij}\cdot |G_{ij}|$, where $G_{ij}$ is a Gaussian random variable have the same distribution. By Jensen's inequality and the fact that $\EE{|G_{ij}|}=\sqrt{2/\pi}$, we get that $\sqrt{2/\pi} \EE{\norm{RA}} \leq \EE{\norm{GA}}/\sqrt{t}$.
Define the function $f:{\{\pm 1\}}^{t\times n} \to \RR$ by $f(S) = \norm{\frac1{\sqrt{t} } SA}$. The calculation above shows that $\text{median}(f)\leq \sqrt{2\pi }$. Since $f$ is convex and $(1/\sqrt{t})$-Lipschitz as a function of the entries of $S$, Talagrand's measure concentration inequality for convex functions yields
\begin{equation*}
 \Prob{ \norm{ RA } \geq \text{median}(f) +\delta} \leq 2 \exp (-\delta^2 t/2).
\end{equation*}
Setting $\delta =1 $ in the above inequality implies the lemma.
\end{proof}
Now using the above Lemma together with Theorem~\ref{thm:matrixmult} (\textit{i.a}) and a simple truncation argument we can prove part (\textit{i.b}).
%%%%%%%%%%%%%%%%%%%%%%%%%%%%%%%%%%%%%%%%%%%%%%%%%%%%%%%%%%%%%%%%%%%%%%%%%
%%%%%%%%%%%%%%%%%%%%%%%%%%%%%%%%%%%%%%%%%%%%%%%%%%%%%%%%%%%%%%%%%%%%%%%%%
\begin{proof}(of Theorem~\ref{thm:matrixmult} (\textit{i.b}))
 Without loss of generality assume that $\norm{A}=\norm{B}=1$. Set $r =\lfloor  \frac{1600 \max\{ \sr{A}, \sr{B} \}}{\eps^2}\rfloor$. Set $\widehat{A} = A - A_r$, $\widehat{B} = B- B_r$. Since $\frobnorm{A}^2 = \sum_{j=1}^{\rank{A}} \sigma_j(A)^2$,
\begin{eqnarray*}
 	\norm{\widehat{A}} \ \leq\ \dfrac{\frobnorm{A} }{\sqrt{r}} \leq \dfrac{\eps}{40}, \mbox{ and } \norm{\widehat{B}}	\ \leq \ \dfrac{\frobnorm{B} }{\sqrt{r}} \leq \dfrac{\eps}{40}.
\end{eqnarray*}
By triangle inequality, it follows that
\begin{eqnarray}
 	\norm{ \widetilde{A}^\top \widetilde{B} - A^\top B} \nonumber \\
						    & \leq & \norm{ A_r^\top R^\top  R B_r - A_r^\top B_r} \label{ineq:rud1}\\
						    &   +  & \norm{ \widehat{A}^\top R^\top R B_r}  \nonumber \\
						    &   +  & \norm{ A_r^\top R^\top  R \widehat{B} } + \norm{ \widehat{A}^\top R^\top R \widehat{B}} \label{ineq:rud2}\\
						    &   +  & \norm{ \widehat{A}^\top B_r} + \norm{ A_r^\top \widehat{B} } + \norm{ \widehat{A}^\top \widehat{B} }\label{ineq:rud3}.
\end{eqnarray}
Choose a constant in Theorem~\ref{thm:matrixmult} (\textit{i.a}) so that the failure probability of the right hand side of~\eqref{ineq:rud1} does not exceed $\exp (-c\eps^2 t)$, where $c=c_1/32$. The same argument shows that $\Prob{ \norm{R A_r} \geq 1 + \eps } \leq \exp( -c\eps^2 t)$ and $\Prob{ \norm{ R B_r} \geq 1 + \eps } \leq \exp( -c\eps^2 t)$. This combined with Lemma~\ref{lem:Rudelson} applied on $\widehat{A}$ and $\widehat{B}$ yields that the sum in~\eqref{ineq:rud2} is less than $2(1+\eps) \eps / 10 +\eps^2 /100$. Also, since $\norm{A_r},\norm{B_r} \leq 1$, the sum in~\eqref{ineq:rud3} is less that $2\eps/10 + \eps^2 /100$. Combining the bounds for~\eqref{ineq:rud1},~\eqref{ineq:rud2} and \eqref{ineq:rud3} concludes the claim.
\end{proof}
%%%%%%%%%%%%%%%%%%%%%%%%%%%%%%%%%%%%%%%%%%%%%%%%%%%%%%%%%%%%%%%%%%%
%%%%%%%%%%%%%%%%%%%%%%%%%%%%%%%%%%%%%%%%%%%%%%%%%%%%%%%%%%%%%%%%%%%
\paragraph{Row Sampling - Part (\textit{ii}):}
%%%%%%%%%%%%%%%%%%%%%%%%%%%%%%%%%%%%%%%%%%%%%%%%%%%%%%%%%%%%%%%%%%%
%%%%%%%%%%%%%%%%%%%%%%%%%%%%%%%%%%%%%%%%%%%%%%%%%%%%%%%%%%%%%%%%%%%
By homogeneity normalize $A$ and $B$ such that  $\norm{A}=\norm{B}=1$. Notice
that $A^\top B = \sum_{i=1}^{n} A_{(i)}^\top B_{(i)}$. Define $p_i =
\frac{\norm{A^\top_{(i)}}\norm{B_{(i)}} }{S}$, where
$S=\sum_{i=1}^{n}{\norm{A_{(i)}^\top}\norm{B_{(i)}}}$. Also define a distribution over
 matrices in $\RR^{(m+p)\times (m+p)}$ with $n$ elements by
\[\Prob{M=\frac1{p_i}\left[\begin{array}[c]{ll}
0 & B^\top_{(i)} A_{(i)} \\ 
A^\top_{(i)} B_{(i)} & 0
\end{array}
\right]} = p_i.\]
First notice that 
\begin{eqnarray*}
 \EE{M}  &  =  &  \sum_{i=1}^{n}{\frac1{p_i}\left[\begin{array}[c]{ll}
0 & B^\top_{(i)} A_{(i)} \\ 
A^\top_{(i)} B_{(i)} & 0
\end{array}
\right]}\cdot p_i \\
         &  =  &  \sum_{i=1}^{n}{\left[\begin{array}[c]{ll}
0 & B^\top_{(i)} A_{(i)} \\ 
A^\top_{(i)} B_{(i)} & 0
\end{array}
\right]} \\
	&  =  &  
\left[\begin{array}[c]{ll}
0 & B^\top A \\ 
A^\top B & 0
\end{array}
\right].
\end{eqnarray*}
This implies that $\norm{\EE{M}}= \norm{A^\top B} \leq 1$. Next notice that the spectral norm of the random matrix $M$ is upper bounded by
$\sqrt{\sr{A}\sr{B} }$ almost surely. Indeed,
\begin{eqnarray*}
 \norm{M}    		&   \leq  &  \sup_{i\in{[n]}}\norm{\dfrac{A^\top_{(i)} B_{(i)}}{p_i}}\\
			&   =  &  S\sup_{i\in{[n]}} \norm{\dfrac{A_{(i)}^\top}{\norm{A_{(i)}} } \dfrac{B_{(i)}}{\norm{B_{(i)}} }} = S \cdot 1\\
		 			&   =  &  \sum_{i=1}^{n}{\norm{A_{(i)}}\norm{B_{(i)}} } 
		 			\ \leq  \  \frobnorm{A}\frobnorm{B} \\
		 			&   =   & \sqrt{\sr{A}\sr{B}} 
					\ \leq  \ (\sr{A} + \sr{B} ) /2,
\end{eqnarray*}
by definition of $p_i$, properties of norms, Cauchy-Schwartz inequality, and arithmetic/geometric mean inequality. Notice that this quantity (since the spectral norms of both $A,B$ are one) is at most $\widetilde{r}$ by assumption. Also notice that every element on the support of the random variable $M$, has rank at most two. It is easy to see that, by setting $\gamma = \widetilde{r}$, all the conditions in Theorem~\ref{thm:chernoff:matrix_valued:low_rank} are satisfied, and hence we get $i_1,i_2, \dots ,
i_t$ indices from $[n]$, $t=\Omega(\widetilde{r} \log (\widetilde{r}/\eps^2) /\eps^2 ) $, such that with high probability
\begin{eqnarray*}
\|\frac1{t} 
\sum_{j=1}^{t}{\left[\begin{array}[c]{ll}
0 & \frac1{p_{i_j}} B^\top_{(i_j)} A_{(i_j)} \\ 
\frac1{p_{i_j}}A^\top_{(i_j)} B_{(i_j)} & 0
\end{array}
\right]}\\
 - 
\left[\begin{array}[c]{ll}
0 & B^\top A \\ 
A^\top B & 0
\end{array}
\right]\|_2 &\leq& \eps.
\end{eqnarray*}
The first sum can be rewritten as $\widetilde{A}^\top \widetilde{B}$ where $\widetilde{A}
=\frac1{\sqrt{t}}
\left[\begin{array}[l]{llll}
\frac1{\sqrt{p_{i_1}}}A_{(i_1)}^\top & \frac1{\sqrt{p_{i_2}}}A_{(i_2)}^\top & \dots
& \frac1{\sqrt{p_{i_t}}}A_{(i_t)}^\top
\end{array}
\right]^\top$ and $ \widetilde{B} = \frac1{\sqrt{t}} \left[\begin{array}[l]{llll}
\frac1{\sqrt{p_{i_1}}}B_{(i_1)}^\top & \frac1{\sqrt{p_{i_2}}}B_{(i_2)}^\top & \dots
& \frac1{\sqrt{p_{i_t}}}B_{(i_t)}^\top
\end{array}
\right]^\top$. This concludes the theorem.
%%%%%%%%%%%%%%%%%%%%%%%%%%%%%%%%%%%%%%%%%%%%%%%%%%%%%%%%%%%%%%%%%%%
%%%%%%%%%%%%%%%%%%%%%%%%%%%%%%%%%%%%%%%%%%%%%%%%%%%%%%%%%%%%%%%%%%%
\subsection{Proof of Theorem~\ref{thm:ell2_regression} ($\ell_2$-regression)}
%%%%%%%%%%%%%%%%%%%%%%%%%%%%%%%%%%%%%%%%%%%%%%%%%%%%%%%%%%%%%%%%%%%
%%%%%%%%%%%%%%%%%%%%%%%%%%%%%%%%%%%%%%%%%%%%%%%%%%%%%%%%%%%%%%%%%%%
\begin{proof}(of Theorem~\ref{thm:ell2_regression})
Similarly as the proof in~\cite{sarlos}. Let $A=U\Sigma V^\top$ be the SVD of $A$. Let  $b=Ax_{opt} + w$, where $w\in\RR^n$ and $w\bot$\text{colspan(A)}. Also let $A(\widetilde{x}_{opt} - x_{opt})=Uy$, where $y \in \RR^{\rank{A}}$. Our goal is to bound this quantity
\begin{eqnarray}
 \norm{b-A\widetilde{x}_{opt}}^2 &=& \norm{b-A(\widetilde{x}_{opt} - x_{opt}) -
Ax_{opt}}^2\nonumber \\
                            &=& \norm{w - Uy}^2 \nonumber \\
			    & = & \norm{w}^2 +\norm{Uy}^2,  \quad \text{since }w\bot
\text{colspan}(U)\nonumber \\
				& = & \norm{w}^2 + \norm{y}^2, \quad \text{since
}U^\top U = I. \label{eq:l2_basic}
\end{eqnarray}
It suffices to bound the norm of $y$, i.e., $\norm{y} \leq 3\eps \norm{w}$. Recall that given $A,b$ the vector $w$ is uniquely defined. On the other hand, vector $y$ depends on the random projection $R$. Next we show the connection between $y$ and $w$ through the ``normal equations''.
\begin{eqnarray}
	RA\widetilde{x}_{opt}  &=& Rb +w_2 \implies \nonumber \\
	RA\widetilde{x}_{opt}  &=& R(Ax_{opt} + w) +w_2 \implies \nonumber \\
	RA(\widetilde{x}_{opt} - x_{opt})  &=& Rw + w_2 \implies \nonumber \\ 
	U^\top R^\top R U y  &=& U^\top R^\top Rw  + U^\top  R^\top  w_2 \implies
\nonumber \\
	U^\top R^\top R U y  &=& U^\top R^\top Rw \label{eq:random_l2},
\end{eqnarray}
where $w_2\bot\text{colspan}(R)$, and used this fact to derive Ineq.~\eqref{eq:random_l2}. A crucial observation is that the  $\text{colspan}(U)$ is perpendicular to $w$. Set $A=B=U$ in Theorem~\ref{thm:matrixmult}, and set $\eps' = \sqrt{\eps}$, and $t=\Omega( r /\eps'^2)$. Notice that $\rank{A}+\rank{B} \leq 2r$, hence with constant probability we know that $1-\eps' \leq \sigma_i(RU) \leq 1 + \eps'$. It follows that $\norm{U^\top R^\top R U y } \geq (1-\eps')^2 \norm{y}$. A similar argument (set $A=U$ and $B=w$ in Theorem~\ref{thm:matrixmult}) guarantees that $ \norm{U^\top R^\top Rw } = \norm{U^\top R^\top Rw - U^\top w} \leq \eps' \norm{U} \norm{w}= \eps' \norm{w}$. Recall that $\norm{U}  = 1$, since $U^\top U = I_n$ with high probability. Therefore, taking Euclidean norms on both sides of Equation~\eqref{eq:random_l2} we get that 
\[ \norm{ y} \leq \dfrac{\eps'}{(1-\eps')^2} \norm{w} \leq 4\eps' \norm{w}. \] 
Summing up, it follows from Equation~\eqref{eq:l2_basic} that, with constant probability, $\norm{b-A\widetilde{x}_{opt}}^2 \leq (1+16\eps'^2) \norm{b-Ax_{opt}}^2= (1 + 16\eps) \norm{b-Ax_{opt}}^2.$ This proves Ineq.~\eqref{ineq:regression:approx}.

Ineq.~\eqref{ineq:regression:x_opt} follows directly from the bound on the norm of $y$ repeating the above proof for $\eps' \leftarrow \eps $. First recall that $x_{opt}$ is in the row span of $A$, since $x_{opt} = V\Sigma^{-1} U^\top b$ and the columns of $V$ span the row space of $A$. Similarly for $\widetilde{x}_{opt}$ since the row span of $R\cdot A$ is contained in the row-span of $A$. Indeed, $\eps \norm{w} \geq \norm{y} =\norm{Uy} = \norm{A(x_{opt} - \widetilde{x}_{opt}) } \geq \sigma_{min(A)} \norm{ x_{opt} - \widetilde{x}_{opt} }$.
\end{proof}
%%%%%%%%%%%%%%%%%%%%%%%%%%%%%%%%%%%%%%%%%%%%%%%%%%%%%%%%%%%%%%%%%%%
%%%%%%%%%%%%%%%%%%%%%%%%%%%%%%%%%%%%%%%%%%%%%%%%%%%%%%%%%%%%%%%%%%%
\subsection{Proof of Theorems~\ref{thm:lowrank},~\ref{thm:lowrank:low_stable_tail} (Spectral Low Rank Matrix Approximation)}\label{sec:app:spectral}
%%%%%%%%%%%%%%%%%%%%%%%%%%%%%%%%%%%%%%%%%%%%%%%%%%%%%%%%%%%%%%%%%%%
%%%%%%%%%%%%%%%%%%%%%%%%%%%%%%%%%%%%%%%%%%%%%%%%%%%%%%%%%%%%%%%%%%%
%%%%%%%%%%%%%%%%%%%%%%%%%%%%%%%%%%%%%%%%%%%%%%%%%%%%%%%%%%%%%%%%%%%
%%%%%%%%%%%%%%%%%%%%%%%%%%%%%%%%%%%%%%%%%%%%%%%%%%%%%%%%%%%%%%%%%%%
%%%%%%%%%%%%%%%%%%%%%%%%%%%%%%%%%%%%%%%%%%%%%%%%%%%%%%%%%%%%%%%%%%%
%%%%%%%%%%%%%%%%%%%%%%%%%%%%%%%%%%%%%%%%%%%%%%%%%%%%%%%%%%%%%%%%%%%
\begin{proof}(of Lemma~\ref{lem:rayleigh_implies_lowrank})
By the assumption and using Lemma~\ref{lem:rayleight_to_eig} we get that
\begin{equation}\label{eq:sigma}
 (1-\eps) \sigma_i(A^\top A) \leq \sigma_i(\widetilde{A}^\top \widetilde{A}) \leq
(1+\eps) \sigma_i(A^\top A)
\end{equation}
for all $i=1,\ldots ,\rank{A}$. Let $\widetilde{\Pi}_k$ be the projection matrix onto the first $k$ right singular vectors of $\widetilde{A}$, i.e., $\pinv{(\widetilde{A}_k)}\widetilde{A}_k$. It follows that for every $k=1,\dots ,\rank{A}$
\begin{eqnarray*}
\norm{A- P_{\widetilde{A},k}(A)} & \leq & \norm{A-A \widetilde{\Pi}_k}^2 \\
			     &   =   & \sup_{x\in\RR^m,\ \|x\|=1}{\norm{A (I-\widetilde{\Pi}_k)x}^2}\\
			     &   =   & \sup_{x\in \ker \widetilde{\Pi}_k,\ \|x\|=1}{\|A x\|_2^2} \\
			     &   =   & \sup_{x\in \ker \widetilde{\Pi}_k,\ \|x\|=1}{ x^\top A^\top A x }\\
			     & \leq\ & (1+\eps) \sup_{x\in \ker \widetilde{\Pi}_k,\ \norm{x}=1} x^\top \widetilde{A}^\top \widetilde{A} x\\
			     &   =   & (1+\eps) \sigma_{k+1}(\widetilde{A }^\top \widetilde{A }) \\
		  	     & \leq  & (1+\eps)^2 \sigma_{k+1}(A^\top A)\\
%			     &   =   & (1+\eps)^2 \sigma^2_{k+1}(A)\\
			     &   =   & (1+\eps)^2 \norm{A-A_k}^2,
\end{eqnarray*}
using that $x\bot \ker{\widetilde{\Pi}_k}$ implies $\widetilde{\Pi}_k x = x$, left side of the hypothesis, Courant-Fischer on $\widetilde{A}^\top \widetilde{A}$ (see Eqn.~\eqref{eqn:Courant_Fischer}), Eqn.~\eqref{eq:sigma}, and properties of singular values, respectively.
\end{proof}
%%%%%%%%%%%%%%%%%%%%%%%%%%%%%%%%%%%%%%%%%%%%%%%%%%%%%%%%%%%%%%%%%%%
%%%%%%%%%%%%%%%%%%%%%%%%%%%%%%%%%%%%%%%%%%%%%%%%%%%%%%%%%%%%%%%%%%%
\paragraph{Proof of Theorem~\ref{thm:lowrank} (\textit{i}):}
%%%%%%%%%%%%%%%%%%%%%%%%%%%%%%%%%%%%%%%%%%%%%%%%%%%%%%%%%%%%%%%%%%%
%%%%%%%%%%%%%%%%%%%%%%%%%%%%%%%%%%%%%%%%%%%%%%%%%%%%%%%%%%%%%%%%%%%
\paragraph{Part (\textit{a}):}
Now we are ready to prove our first corollary of our matrix multiplication result to the problem of computing an approximate low rank matrix approximation of a matrix with respect to the spectral norm (Theorem~\ref{thm:lowrank}). 
\begin{proof}
Set $\widetilde{A} = \frac1{\sqrt{t}}RA$ where $R$ is a $\Omega(r/\eps^2)\times n$
random sign matrix. Apply Theorem~\ref{thm:matrixmult} \textit{(i.a)} on $A$ we have with high probability that
\begin{equation}\label{ineq:rayleigh:low_rank}
\forall~x\in\RR^n,\ (1-\eps)  x^\top A^\top A x \leq x^\top \widetilde{A}^\top
\widetilde{A}x \leq (1+\eps)x^\top A^\top Ax.
\end{equation}
Combining Lemma~\ref{lem:rayleigh_implies_lowrank} with Ineq.~\eqref{ineq:rayleigh:low_rank} concludes the proof.
\end{proof}
\paragraph{Part (\textit{b}):}
The proof is based on the following lemma which reduces the problem of low rank matrix approximation to the problem of bounding the norm of a random matrix. We restate it here for reader's convenience and completeness~\cite[Lemma~8]{low_rank:STOC09}, (see also~\cite[Theorem~$9.1$]{matrix:survey:HMT09} or \cite{cssp:boutsidis}).
\begin{lemma}\label{lem:low_rank:STOC09}
Let $A=A_k + U_{r-k}\Sigma_{r-k}V_{r-k}^\top$, $H_k = U_{r-k}\Sigma_{r -k}$ and
$R$ be \emph{any} $t\times n$ matrix. If the matrix $(RU_k)$ has full column
rank, then the following inequality holds,
\begin{equation}
	\norm{A-P_{(RA),k}(A)} \leq 2 \norm{A-A_k}~+~\norm{\pinv{(RU_k)}RH_k}.
\end{equation}
\end{lemma}
Notice that the above lemma, reduces the problem of spectral low rank matrix approximation to a problem of approximation the spectral norm of the random
matrix $\pinv{(RU_k)}RH_k$. 

First notice that by setting $t=\Omega(k/\eps^2)$ we can guarantee that the matrix $(RU_k)$ will have full column rank with high probability. Actually, we can say something much stronger; applying Theorem~\ref{thm:matrixmult} (\textit{i.a}) with $A=U_k$ we can guarantee that all the singular values are within $1\pm \eps$ with high probability. Now by conditioning on the above event ( $(RU_k)$ has full column rank), it follows from Lemma~\ref{lem:low_rank:STOC09} that
\begin{eqnarray*}
	\norm{A-P_{(RA),k}(A)}  & \leq & 2\norm{A-A_k} + \norm{\pinv{(RU_k)}RH_k}\\
				& \leq & 2\norm{A-A_k} + \norm{\pinv{(RU_k)}}\norm{RH_k} \\
				& \leq & 2\norm{A-A_k} + \frac1{1-\eps}\norm{RH_k}\\
				& \leq & 2\norm{A-A_k} + \frac{3}{2}\norm{RU_{r-k}}\norm{\Sigma_{r-k}} 
\end{eqnarray*}
using the sub-multiplicative property of matrix norms, and that $\eps <1/3$. Now, it suffices to bound the norm of $W:=RU_{r-k}$. Recall that
$R=\frac1{\sqrt{t}} G$ where $G$ is a $t\times n$ random Gaussian matrix, It is well-known that the distribution of the random matrix $GU_{r-k}$ (by rotational invariance of the Gaussian distribution) has entries which are also i.i.d. Gaussian random variables.
\ignore{
\begin{claim}\label{claim:rot_invariance}
Let $G$ be a $t\times n$ random zero-mean sub-Gaussian matrix, then $W=GU_{r-k}$ is a $t\times (r - k)$ random zero-mean sub-Gaussian matrix.
\end{claim}
\begin{proof}
Let $P=U_{r-k}$. To see this argument, note that any linear (fixed) combination of sub-Gaussian random variables is sub-Gaussian (with different sub-Gaussian constant). Now by the linearity of expectation we can easily show that every entry of $GP$ has expected value zero. Moreover, the correlation between two entries $E[(GP)_{ij}(GP)_{lk}]=E[(\sum_{r=1}^{t}{G_{ir}P_{rj}})$ $(\sum_{r=i}^{t}{G_{lr}P_{rk}})]$ is zero if $i\neq l$, and is equal to the inner product of the $j^{th}$ and $k^{th}$ column of $P$ otherwise. This gives that the covariance matrix is\footnote{For matrices $A,B$, we denote $A\otimes B$ the Kronecker product between them.} $I_{t} \otimes U_{r-k}^\top U_{r-k}=I_t\otimes
I_{r-k}=I_{t (r-k)}$, which implies that the entries of $W$ are i.i.d..
\end{proof}
}
Now, we can use the following fact about random sub-Gaussian matrices to give a bound on the spectral norm of $W$. Indeed, we have the following
\begin{theorem}\cite[Proposition~2.3]{rand_matrix:VR:l2_norm_rectangular}\label{thm:subgaussian_norm}
Let $W$ be a $t\times {(r-k)}$ random matrix whose entries are independent mean zero Gaussian random variables. Assume that $r-k\geq t$, then
\begin{equation}
	\Prob{\norm{W} \geq \delta \sqrt{r-k} } \leq e^{-c_0\delta^2\sqrt{r-k}}.
\end{equation}
for any $\delta >\delta_0$, where $\delta_0$ is a positive constant.
\end{theorem}
Apply union bound on the above theorem with $\delta$ be a sufficient large constant and on the conditions of Lemma~\ref{lem:low_rank:STOC09}, we get that with high probability, $\norm{W} \leq C_3\sqrt{r-k}$ \emph{and} $\sigma_{\min}(\pinv{(RU_k)}) \leq 1/(1-\eps)$. Hence, Lemma~\ref{thm:subgaussian_norm} combined with the above discussion implies that
\clearpage
\begin{eqnarray*}
	\norm{A-P_{(RA),k}(A)}  & \leq & 2\norm{A-A_k}\\
				&   +  & 3/2\norm{RU_{r-k}}\norm{A-A_k} \\
				&   =  & 2\norm{A-A_k} \\
				& + &  \frac{3}{2\sqrt{t}}\norm{GU_{r-k}}\norm{A-A_k} \\
				%& \leq &  2\norm{A-A_k} +\frac{3C_3\sqrt{r-k}}{2\sqrt{t}}\norm{A-A_k} \\
				& \leq & \left(2 + c_4\eps\sqrt{\frac{r-k}{k}}\right) \norm{A-A_k},
\end{eqnarray*}
where $c_4>0$ is an absolute constant. Rescaling $\eps$ by $c_4$ concludes
Theorem~\ref{thm:lowrank} (\textit{i.b}).
%%%%%%%%%%%%%%%%%%%%%%%%%%%%%%%%%%%%%%%%%%%%%%%%%%%%%%%%%%%%%%%%%%%
%%%%%%%%%%%%%%%%%%%%%%%%%%%%%%%%%%%%%%%%%%%%%%%%%%%%%%%%%%%%%%%%%%%
\paragraph{Proof of Theorem~\ref{thm:lowrank} (\textit{ii})}
%%%%%%%%%%%%%%%%%%%%%%%%%%%%%%%%%%%%%%%%%%%%%%%%%%%%%%%%%%%%%%%%%%%
%%%%%%%%%%%%%%%%%%%%%%%%%%%%%%%%%%%%%%%%%%%%%%%%%%%%%%%%%%%%%%%%%%%
Here we prove that we can achieve the same relative error bound as with random projections by just sampling rows of $A$ through a judiciously selected distribution. However, there is a price to pay and that's an extra logarithmic factor on the number of samples, as is stated in
Theorem~\ref{thm:lowrank}, part (\textit{ii}). 
\begin{proof}(of Theorem~\ref{thm:lowrank} (\textit{ii}))
The proof follows closely the proof of~\cite{graph:sparsifiers:eff_resistance}. Similar with the proof of part (\textit{a}).  Let $A=U\Sigma V^\top $ be the singular value decomposition of $A$. Define the projector matrix $\Pi = U U^\top$ of size $n\times n$. Clearly, the rank of $\Pi$ is equal to the rank of $A$ and $\Pi$ has the same image with $A$ since every element in the image of $A$ and $\Pi$ is a linear combination of columns of $U$. Recall that for any projection matrix, the following holds $\Pi^2=\Pi$ and hence $\sr{\Pi}=\rank{A}=r$. Moreover, $\sum_{i=1}^{n}\norm{U_{(i)}}^2=\trace{UU^\top}=\trace{\Pi}=\trace{\Pi^2}=r$. Let $p_i = \Pi(i,i)/r=\norm{U_{(i)}}^2/r$ be a probability distribution on $[n]$, where $U_i$ is
the $i$-th row of $U$.

Define a $t\times n$ random matrix $S$ as follows: Pick $t$ samples from $p_i$; if the $i$-th sample is equal to $j(\in{[n]})$ then set $S_{ij} = 1/\sqrt{p_j}$. Notice that $S$ has exactly one non-zero entry in each row, hence it has $t$ non-zero entries. Define $\widetilde{A} = S A$.

It is easy to verify that $\EE_S{\Pi S^\top S \Pi} = \Pi^2 =\Pi$. Apply Theorem~\ref{thm:chernoff:matrix_valued:low_rank} (alternatively we can use \cite[Theorem~3.1]{lowrank:rankone:VR}, since the matrix samples are rank one) on the matrix $\Pi$, notice that $\frobnorm{\Pi}^2=r$ and $\norm{\Pi}=1$, $\norm{\EE_S{\Pi S^\top S \Pi }}\leq 1$, hence the stable rank of $\Pi$ is $r$. Therefore, if $t=\Omega (r \log (r/\eps^2) /\eps^2)$ then with high probability 
\begin{equation}\label{ineq:projector_approx_implies_rayleigh}
\norm{\Pi S^\top S\Pi -\Pi \Pi} \leq \eps.
\end{equation}
It suffices to show that Ineq.~\eqref{ineq:projector_approx_implies_rayleigh} is equivalent with the condition of Lemma~\ref{lem:rayleigh_implies_lowrank}. Indeed,
\begin{eqnarray*}
  \sup_{x\in\RR^n,~x\neq 0} \left|\frac{x^\top (\Pi S^\top S\Pi - \Pi \Pi)
x}{x^\top x}\right| \leq \eps & \Leftrightarrow & \\
\sup_{x\not{\in{\ker{\Pi}}},~x\neq 0} \frac{\left|x^\top (\Pi S^\top S\Pi - \Pi \Pi) x\right|}{x^\top x}
\leq \eps & \Leftrightarrow & \\
%\sup_{x\in{\text{Im}(\Pi)},~x\neq 0 } \frac{\left|x^\top (\Pi S^\top S \Pi - \Pi \Pi) x\right|}{x^\top x} \leq \eps & \Leftrightarrow & \\
\sup_{y\in{\text{Im}(A)},~y\neq 0} \frac{\left|y^\top (\Pi S^\top S\Pi - \Pi
\Pi) y\right|}{y^\top y} \leq \eps & \Leftrightarrow & \\
\sup_{x\in\RR^m,~A x\neq 0}\frac{\left|x^\top A^\top (\Pi S^\top S\Pi - \Pi \Pi)A x\right|}{x^\top A^\top A
x} \leq \eps 	& \Leftrightarrow & \\
\sup_{x\in\RR^m,~Ax\neq 0}
\frac{\left|x^\top (A^\top S^\top S A - A^\top A)x\right|}{x^\top A^\top A x}
\leq \eps & \Leftrightarrow & \\
\sup_{x\in\RR^m,~Ax\neq 0}\frac{\left|x^\top (\widetilde{A}^\top \widetilde{A} - A^\top A)x\right|}{x^\top A^\top
A x} \leq \eps,
\end{eqnarray*}
since $x\not\in{\ker{\Pi}}$ implies $x\in{\im{A}}$, $\im{A}\equiv \im{\Pi}$, and $\Pi A = A$. By re-arranging terms we get Equation~\eqref{ineq:rayleigh:low_rank} and so the claim follows.
\end{proof}
%%%%%%%%%%%%%%%%%%%%%%%%%%%%%%%%%%%%%%%%%%%%%%%%%%%%%%%%%%%%%%%%%%%
%%%%%%%%%%%%%%%%%%%%%%%%%%%%%%%%%%%%%%%%%%%%%%%%%%%%%%%%%%%%%%%%%%%
\paragraph{Proof of Theorem~\ref{thm:lowrank:low_stable_tail}:}
%%%%%%%%%%%%%%%%%%%%%%%%%%%%%%%%%%%%%%%%%%%%%%%%%%%%%%%%%%%%%%%%%%%
%%%%%%%%%%%%%%%%%%%%%%%%%%%%%%%%%%%%%%%%%%%%%%%%%%%%%%%%%%%%%%%%%%%
Similarly with the proof of Theorem~\ref{thm:lowrank} (\textit{i.b}). By following the proof of part (\textit{i.b}), conditioning on the event that $(RU_k)$ has full column rank in Lemma~\ref{lem:low_rank:STOC09}, we get with high probability that
\begin{eqnarray*}
 \norm{A-P_{\widetilde{A},k}(A)}  & \leq & 2\norm{A-A_k} + \frac{\norm{U_k^\top R^\top RH_k}}{(1-\eps)^2}
\end{eqnarray*}
using the fact that if $(RU_k)$ has full column rank then $\pinv{(RU_k)} = ((RU_k)^\top R U_k)^{-1} U_k^\top R^\top $ and $\norm{((RU_k)^\top RU_k)^{-1}} \leq 1/(1-\eps)^2$. Now observe that $U_k^\top H_k=0$. Since $\sr{H_k} \leq k$, using Theorem~\ref{thm:matrixmult} (\textit{i.b}) with $t=\Omega(k/\eps^4)$, we get that $\norm{ U_k^\top R^\top R H_k } = \norm{ U_k^\top R^\top RH_k - U_k^\top H_k} \leq \eps \norm{U_k} \norm{H_k} = \eps \norm{A-A_k}$ with high probability. Rescaling $\eps$ concludes the proof.
\section{Acknowledgments}
Many thanks go to Petros Drineas for many helpful discussions and pointing out the connection of Theorem~\ref{thm:matrixmult} with the $\ell_2$-regression problem. The second author would like to thank Mark Rudelson for his valueable comments on an earlier draft and also for sharing with us the proof of~Theorem~\ref{thm:matrixmult} (\textit{i.b}).
\clearpage
%%%%%%%%%%%%%%%%%%%%%%%%%%%%%%%%%%%%%%%%%%%%%%%%%%%%%%%%%%%%%%%%%%%
%%%%%%%%%%%%%%%%%%%%%%%%%%%%%%%%%%%%%%%%%%%%%%%%%%%%%%%%%%%%%%%%%%%
%	Bibliography
%%%%%%%%%%%%%%%%%%%%%%%%%%%%%%%%%%%%%%%%%%%%%%%%%%%%%%%%%%%%%%%%%%%
%%%%%%%%%%%%%%%%%%%%%%%%%%%%%%%%%%%%%%%%%%%%%%%%%%%%%%%%%%%%%%%%%%%
{\tiny
\bibliographystyle{alpha}

\begin{thebibliography}{DRVW06}

\bibitem[AHK06]{matrix:sparsification:arora}
S.~Arora, E.~Hazan, and S.~Kale.
\newblock {A Fast Random Sampling Algorithm for Sparsifying Matrices}.
\newblock In {\em Proceedings of the International Workshop on Randomization
  and Approximation Techniques (RANDOM)}, pages 272--279, 2006.

\bibitem[AM07]{matrix:sparsification:optas}
D.~Achlioptas and F.~Mcsherry.
\newblock {Fast Computation of Low-rank Matrix Approximations}.
\newblock {\em Journal of the ACM (JACM)}, 54(2):9, 2007.

\bibitem[AW02]{chernoff:matrix_valued:AW}
R.~Ahlswede and A.~Winter.
\newblock {Strong Converse for Identification via Quantum Channels}.
\newblock {\em IEEE Transactions on Information Theory}, 48(3):569--579, 2002.

\bibitem[Bha96]{book:matrix:Bhatia}
R.~Bhatia.
\newblock {\em Matrix Analysis}, volume 169.
\newblock {Graduate Texts in Mathematics, Springer}, {First} edition, 1996.

\bibitem[BMD09]{cssp:boutsidis}
C.~Boutsidis, M.~W. Mahoney, and P.~Drineas.
\newblock {An Improved Approximation Algorithm for the Column Subset Selection
  Problem}.
\newblock In {\em Proceedings of the ACM-SIAM Symposium on Discrete Algorithms
  (SODA)}, pages 968--977, 2009.

\bibitem[Buc01]{khintchine:Buchholz}
A.~Buchholz.
\newblock {Operator Khintchine Inequality in Non-commutative Probability}.
\newblock {\em {Mathematische Annalen}}, 319(1-16):1--16, January 2001.

\bibitem[Cla08]{jl:manifold}
K.~L. Clarkson.
\newblock {Tighter Bounds for Random Projections of Manifolds}.
\newblock In {\em Proceedings of the ACM Symposium on Computational Geometry
  (SoCG)}, pages 39--48, 2008.

\bibitem[CR07]{chernoff:matrix_valued:Candes_Sparse}
E.~Cand\`{e}s and J.~Romberg.
\newblock {Sparsity and Incoherence in Compressive Sampling}.
\newblock {\em Inverse Problems}, 23(3):969, 2007.

\bibitem[CW09]{CW_stoc09}
K.~L. Clarkson and D.~P. Woodruff.
\newblock {Numerical Linear Algebra in the Streaming Model}.
\newblock In {\em Proceedings of the Symposium on Theory of Computing (STOC)},
  pages 205--214, 2009.

\bibitem[DK03]{lra:drineaskannan}
P.~Drineas and R.~Kannan.
\newblock {Pass Efficient Algorithms for Approximating Large Matrices}.
\newblock In {\em Proceedings of the ACM-SIAM Symposium on Discrete Algorithms
  (SODA)}, pages 223--232, 2003.

\bibitem[DKM06a]{matrixmult:drineas}
P.~Drineas, R.~Kannan, and M.~W. Mahoney.
\newblock {Fast Monte Carlo Algorithms for Matrices I: Approximating Matrix
  Multiplication}.
\newblock {\em Journal of the ACM (JACM)}, 36(1):132--157, 2006.

\bibitem[DKM06b]{lowrank:drineas}
P.~Drineas, R.~Kannan, and M.~W. Mahoney.
\newblock {Fast Monte Carlo Algorithms for Matrices II: Computing a Low-Rank
  Approximation to a Matrix}.
\newblock {\em Journal of the ACM (JACM)}, 36(1):158--183, 2006.

\bibitem[DKM06c]{matrixdecomp:drineas}
P.~Drineas, R.~Kannan, and M.~W. Mahoney.
\newblock {Fast Monte Carlo Algorithms for Matrices III: Computing a Compressed
  Approximate Matrix Decomposition}.
\newblock {\em Journal of the ACM (JACM)}, 36(1):184--206, 2006.

\bibitem[DM10]{matrix:leverage_scores:drineas}
P.~Drineas and M.~W. Mahoney.
\newblock {Effective Resistances, Statistical Leverage, and Applications to
  Linear Equation Solving}.
\newblock Available at~\href{http://arxiv.org/abs/1005.3097}{arxiv:1005.3097},
  May 2010.

\bibitem[DMM06]{l2_regression:drineas06}
P.~Drineas, M.~W. Mahoney, and S.~Muthukrishnan.
\newblock {Sampling Algorithms for $\ell_2$-regression and Applications}.
\newblock In {\em Proceedings of the ACM-SIAM Symposium on Discrete Algorithms
  (SODA)}, pages 1127--1136, 2006.

\bibitem[DMMS09]{l2_regression:drineas:faster}
P.~Drineas, M.~W. Mahoney, S.~Muthukrishnan, and T.~Sarlos.
\newblock {Faster Least Squares Approximation}.
\newblock Available at \href{http://arxiv.org/pdf/0710.1435}{arvix:0710.1435},
  May 2009.

\bibitem[DR10]{matrix:volume_sampling:FOCS2010}
A.~Deshpande and L.~Rademacher.
\newblock {Efficient Volume Sampling for Row/column Subset Selection}.
\newblock In {\em Proceedings of the Symposium on Foundations of Computer
  Science (FOCS)}, 2010.

\bibitem[DRVW06]{lowrank:volume_sampling:DRVW06}
A.~Deshpande, L.~Rademacher, S.~Vempala, and G.~Wang.
\newblock {Matrix Approximation and Projective Clustering via Volume Sampling}.
\newblock In {\em Proceedings of the ACM-SIAM Symposium on Discrete Algorithms
  (SODA)}, pages 1117--1126, 2006.

\bibitem[DZ10]{matrix:sparsification:zouzias}
P.~Drineas and A.~Zouzias.
\newblock {A Note on Element-wise Matrix Sparsification via Matrix-valued
  Chernoff Bounds}.
\newblock Available at~\href{http://arxiv.org/abs/1006.0407}{arxiv:1006.0407},
  June 2010.

\bibitem[FKV04]{lowrank:FKV}
A.~Frieze, R.~Kannan, and S.~Vempala.
\newblock {Fast Monte-carlo Algorithms for Finding Low-rank Approximations}.
\newblock {\em Journal of the ACM (JACM)}, 51(6):1025--1041, 2004.

\bibitem[GLF{\etalchar{+}}09]{chernoff:matrix_valued:Bernstein:Gross}
D.~Gross, Y.-K. Liu, S.~T. Flammia, S.~Becker, and J.~Eisert.
\newblock {Quantum State Tomography via Compressed Sensing}.
\newblock Available at~\href{http://arxiv.org/abs/0909.3304}{0909.3304},
  September 2009.

\bibitem[Gro09]{chernoff:matrix_valued:Gross}
D.~Gross.
\newblock {Recovering Low-rank Matrices from Few Coefficients in any Basis}.
\newblock Available at~\href{http://arxiv.org/abs/0910.1879}{arxiv:0910.1879},
  December 2009.

\bibitem[GV96]{book:GVL}
G.~H. Golub and C.~F. {Van Loan}.
\newblock {\em {Matrix Computations}}.
\newblock {The Johns Hopkins University Press}, {Third} edition, October 1996.

\bibitem[HMT09]{matrix:survey:HMT09}
N.~Halko, P.~G. Martinsson, and J.~A. Tropp.
\newblock {Finding Structure with Randomness: Stochastic Algorithms for
  Constructing Approximate Matrix Decompositions}.
\newblock Available at~\href{http://arxiv.org/abs/0909.4061}{arxiv:0909.4061},
  Sep. 2009.

\bibitem[LP86]{khintchine:LP86}
F.~Lust-Piquard.
\newblock In\'egalit\'es de {K}hintchine dans {$C_p~(1<p<\infty)$}.
\newblock {\em C. R. Acad. Sci. Paris S\'{e}r. I Math.}, 303(7):289--292, 1986.

\bibitem[LPP91]{khintchine:Pisier}
F.~Lust-Piquard and G.~Pisier.
\newblock {Non Commutative Khintchine and Paley Inequalities}.
\newblock {\em Arkiv f\"{o}r Matematik}, 29(1-2):241--260, December 1991.

\bibitem[LT91]{book:LedouxTalagrand}
M.~Ledoux and M.~Talagrand.
\newblock {\em {Probability in {B}anach Spaces}}, volume~23 of {\em {Ergebnisse
  der Mathematik und ihrer Grenzgebiete (3)}}.
\newblock Springer-Verlag, 1991.
\newblock {Isoperimetry and Processes}.

\bibitem[Mag07]{magen07}
A.~Magen.
\newblock {Dimensionality Reductions in $\ell_2$ that Preserve Volumes and
  Distance to Affine Spaces}.
\newblock {\em {Discrete {\&} Computational Geometry}}, 38(1):139--153, 2007.

\bibitem[Mil71]{Dvoretsky:Milman}
V.D. Milman.
\newblock {A new Proof of A. Dvoretzky's Theorem on Cross-sections of Convex
  Bodies}.
\newblock {\em {Funkcional. Anal. i Prilozhen.}}, 5(4):28--37, 1971.

\bibitem[NDT09]{low_rank:STOC09}
N.~H. Nguyen, T.~T. Do, and T.~D. Tran.
\newblock {A Fast and Efficient Algorithm for Low-rank Approximation of a
  Matrix}.
\newblock In {\em Proceedings of the Symposium on Theory of Computing (STOC)},
  pages 215--224, 2009.

\bibitem[Nem07]{chernoff:matrix_valued:opt:Nemirovski}
A.~Nemirovski.
\newblock {Sums of Random Symmetric Matrices and Quadratic Optimization under
  Orthogonality Constraints}.
\newblock {\em Mathematical Programming}, 109(2):283--317, 2007.

\bibitem[Rec09]{recht:simple_completion}
B.~Recht.
\newblock {A Simpler Approach to Matrix Completion}.
\newblock Available
  at~\href{http://arxiv.org/abs/0910.0651v2}{arxiv:0910.0651}, October 2009.

\bibitem[RST09]{low_rank:Tygert2008}
V.~Rokhlin, A.~Szlam, and M.~Tygert.
\newblock {A Randomized Algorithm for Principal Component Analysis}.
\newblock {\em SIAM Journal on Matrix Analysis and Applications},
  31(3):1100--1124, 2009.

\bibitem[Rud99]{rudelson:isotropic}
M.~Rudelson.
\newblock {Random Vectors in the Isotropic Position}.
\newblock {\em J. Funct. Anal.}, 164(1):60--72, 1999.

\bibitem[RV07]{lowrank:rankone:VR}
M.~Rudelson and R.~Vershynin.
\newblock {Sampling from Large Matrices: An Approach through Geometric
  Functional Analysis}.
\newblock {\em Journal of the ACM (JACM)}, 54(4):21, 2007.

\bibitem[RV09]{rand_matrix:VR:l2_norm_rectangular}
M.~Rudelson and R.~Vershynin.
\newblock {The Smallest Singular Value of a Random Rectangular Matrix}.
\newblock {\em Communications on Pure and Applied Mathematics},
  62(1-2):1707--1739, 2009.

\bibitem[Sar06]{sarlos}
T.~Sarlos.
\newblock {Improved Approximation Algorithms for Large Matrices via Random
  Projections}.
\newblock In {\em Proceedings of the Symposium on Foundations of Computer
  Science (FOCS)}, pages 143--152, 2006.

\bibitem[So09a]{chernoff:matrix_valued:opt}
A.~Man-Cho So.
\newblock {Improved Approximation Bound for Quadratic Optimization Problems
  with Orthogonality Constraints}.
\newblock In {\em Proceedings of the ACM-SIAM Symposium on Discrete Algorithms
  (SODA)}, pages 1201--1209, 2009.

\bibitem[So09b]{chernoff:matrix_valued:opt:journal}
A.~Man-Cho So.
\newblock {Moment Inequalities for sums of Random Matrices and their
  Applications in Optimization}.
\newblock {\em Mathematical Programming}, December 2009.

\bibitem[SS90]{book:perturbation:stewart}
G.~W. Stewart and J.~G. Sun.
\newblock {\em {Matrix Perturbation Theory (Computer Science and Scientific
  Computing)}}.
\newblock Academic Press, June 1990.

\bibitem[SS08]{graph:sparsifiers:eff_resistance}
D.~A. Spielman and N.~Srivastava.
\newblock {Graph Sparsification by Effective Resistances}.
\newblock In {\em Proceedings of the Symposium on Theory of Computing (STOC)},
  pages 563--568, 2008.

\bibitem[Tro10]{chernoff:matrix_valued:Tropp}
J.~A. Tropp.
\newblock {User-Friendly Tail Bounds for Sums of Random Matrices}.
\newblock Available at~\href{http://arxiv.org/abs/1004.4389}{arxiv:1004.4389},
  April 2010.

\bibitem[WX08]{chernoff:matrix_valued:derand:WX08}
A.~Wigderson and D.~Xiao.
\newblock {Derandomizing the Ahlswede-Winter Matrix-valued Chernoff Bound using
  Pessimistic Estimators, and Applications}.
\newblock {\em Theory of Computing}, 4(1):53--76, 2008.

\end{thebibliography}
\newcommand{\etalchar}[1]{$^{#1}$}

}
%%%%%%%%%%%%%%%%%%%%%%%%%%%%%%%%%%%%%%%%%%%%%%%%%%%%%%%%%%%%%%%%%%%
%%%%%%%%%%%%%%%%%%%%%%%%%%%%%%%%%%%%%%%%%%%%%%%%%%%%%%%%%%%%%%%%%%%

%%%%%%%%%%%%%%%%%%%%%%%%%%%%%%%%%%%%%%%%%%%%%%%%%%%%%%%%%%%%%%%%%%%
%%%%%%%%%%%%%%%%%%%%%%%%%%%%%%%%%%%%%%%%%%%%%%%%%%%%%%%%%%%%%%%%%%%
\section*{Appendix}
%%%%%%%%%%%%%%%%%%%%%%%%%%%%%%%%%%%%%%%%%%%%%%%%%%%%%%%%%%%%%%%%%%%
%%%%%%%%%%%%%%%%%%%%%%%%%%%%%%%%%%%%%%%%%%%%%%%%%%%%%%%%%%%%%%%%%%%
The next lemma states that if a symmetric positive semi-definite matrix
$\widetilde{A}$ approximates the Rayleigh quotient of a symmetric positive
semi-definite matrix $A$, then the eigenvalues of $\widetilde{A}$ also approximate
the eigenvalues of $A$.
\begin{lemma}\label{lem:rayleight_to_eig}
Let $0<\eps<1$. Assume $A$, $\widetilde{A}$ are $n\times n$ symmetric positive
semi-definite matrices, such 
that the following inequality holds
\[ (1-\eps)x^\top A x \leq x^\top \widetilde{A} x \leq (1+\eps)x^\top A x, \qquad
\forall\ x\in{\RR^n}.\]
Then, for $i=1,\dots, n$ the eigenvalues of $A$ and $\widetilde{A}$ are the same
up-to an error factor $\eps$, i.e.,
\[(1-\eps)\lambda_i(A) \leq \lambda_i(\widetilde{A}) \leq (1+\eps) \lambda_i(A).\]
\end{lemma}
\begin{proof}
The proof is an immediate consequence of the Courant-Fischer's characterization
of the eigenvalues. First notice that by hypothesis, $A$ and $\widetilde{A}$ have
the same null space. Hence we can assume without loss of generality, that
$\lambda_i(A), \lambda_i(\widetilde{A}) > 0$ for all $i=1,\dots, n$.
Let $\lambda_i(A)$ and $\lambda_i(\widetilde{A})$ be the eigenvalues (in non-decreasing order) of
$A$ and $\widetilde{A}$, respectively. The Courant-Fischer min-max
theorem~\cite[p.~394]{book:GVL} expresses the eigenvalues as 
\begin{equation}\label{eqn:Courant_Fischer}
\lambda_i(A) = \min_{S^i}\max_{x\in{S^i} } \frac{x^\top A x}{x^\top x},
\end{equation}
where the minimum is over all $i$-dimensional subspaces $S^{i}$. Let the
subspaces $S^{i}_0$ and $S^i_1$ where the minimum is achieved for the
eigenvalues of $A$ and $\widetilde{A}$, respectively.
Then, it follows that
\[\lambda_i (\widetilde{A}) = \min_{S^i}\max_{x\in{S^i} } \frac{x^\top \widetilde{A} x}{x^\top
x}\leq \max_{x\in{S^i_0} } \frac{x^\top \widetilde{A} x}{x^\top A x} \frac{x^\top A
x}{x^\top x} \leq (1+\eps)\lambda_i (A). \]
and similarly,
\[\lambda_i (A) = \min_{S^i}\max_{x\in{S^i} } \frac{x^\top A x}{x^\top x}\leq
\max_{x\in{S^i_1} } \frac{x^\top A x}{x^\top \widetilde{A} x} \frac{x^\top \widetilde{A}
x}{x^\top x} \leq \frac{\lambda_i (\widetilde{A})}{1-\eps}. \]
Therefore, it follows that for $i=1,\dots , n$,
\[(1- \eps) \lambda_i(A) \leq \lambda_i(\widetilde{A}) \leq (1+\eps) \lambda_i(A).\]
\end{proof}
%%%%%%%%%%%%%%%%%%%%%%%%%%%%%%%%%%%%%%%%%%%%%%%%%%%%%%%%%%%%%%%%%%%
%%%%%%%%%%%%%%%%%%%%%%%%%%%%%%%%%%%%%%%%%%%%%%%%%%%%%%%%%%%%%%%%%%%
\subsection*{Proof of Theorem~\ref{thm:chernoff:matrix_valued:low_rank}}\label{sec:chernoff:matrix_valued:low_rank}
%%%%%%%%%%%%%%%%%%%%%%%%%%%%%%%%%%%%%%%%%%%%%%%%%%%%%%%%%%%%%%%%%%%
%%%%%%%%%%%%%%%%%%%%%%%%%%%%%%%%%%%%%%%%%%%%%%%%%%%%%%%%%%%%%%%%%%%
For notational convenience, let $Z=\norm{ \frac1{t} \sum_{i=1}^{t} M_i -\EE M}$ and define $E_p:= \EE_{M_1,M_2,\dots , M_t}{Z^p}$. Moreover, let
$X_1,X_2,\dots , X_n$ be copies of a (matrix-valued) random variables $X$, we
will denote $\EE_{X_1,X_2,\dots , X_n}$ by $\EE_{X_{[n]}}$. Our goal is to give
sharp bounds on the moments of the non-negative random variable $Z$ and then
using the moment method to give concentration result for $Z$.

First we give a technical lemma of independent interest that bounds the $p$-th
moments of $Z$ as a function of $p$, $r$ (the rank of the samples), and the
$p/2$-th moment of the random variable $\norm{\sum_{j=1}^{t} M_j^2}$. More
formally, we have the following
\begin{lemma}\label{lem:E_p_vs_sum_of_squares}
Let $M_1,\dots , M_t $ be i.i.d. copies of $M$, where $M$ is a symmetric
matrix-valued random variable that has rank at most $r$ almost surely. Then for
every $p\geq 2$
\begin{equation}
 E_p \ \leq \ r t^{1-p} (2B_p)^p \EE_{M_{[t]}}{\norm{\sum_{j=1}^{t}{  M_j^2 }
}^{p/2}},
\end{equation}
where $B_p$ is a constant that depends on $p$.
\end{lemma}
We need a non-commutative version of Khintchine inequality due to F.
Lust-Piquard~\cite{khintchine:LP86}, see also~\cite{khintchine:Pisier} and~\cite[Theorem~5]{khintchine:Buchholz}. We start with some
preliminaries; let $A\in{\RR^{n\times n}}$ and denote by $C_p^n$ the $p$-th
Schatten norm space$-$the Banach space of linear operators (or matrices in our
setting) in $\RR^n-$ equipped with the norm 
\begin{equation}\label{def:Schatten_norm}
	\left\|A\right\|_{\mathrm{C}_p^n}:= \left(\sum_{i=1}^{n}
\sigma_i(A)^p\right)^{1/p},
\end{equation}
where $\sigma_i(A)$ are the singular values of $A$, see~\cite[Chapter~IV,~p.92]{book:matrix:Bhatia} for a discussion on Schatten norms. Notice that $\norm{A} = \sigma_1(A)$, hence we have the following inequality
\begin{equation}\label{ineq:Schatter:rank}
 \norm{A} \leq \left\|A\right\|_{\mathrm{C}_p^n} \leq
\left(\rank{A}\right)^{1/p}\norm{A},
\end{equation}
for any $p \geq 1$. Notice that when $p=\log_2 (\rank{A})$, then $\rank{A}^{1/\log_2 (\rank{A})}=2$.
Therefore, in this case, the Schatten norm is essentially the spectral norm. We
are now ready to state the matrix-valued Khintchine inequality. See
e.g.~\cite{rudelson:isotropic} or~\cite[Lemma~8]{low_rank:STOC09}.
\begin{theorem}\label{thm:non-comm:Khintchine}
Assume $2 \leq p < \infty$. Then there exists a constant $B_p$ such that for any
sequence of $t$ symmetric matrices $M_1,\dots , M_t$, with $M_i\in{C_{p}^{n}}$,
such that the following inequalities hold
\begin{equation}\label{ineq:non_com:Khintchine}
	%\left\|\left(\sum_{i=1}^{t}{M_i^2}\right)^{1/2}\right\|_{\mathrm{C}_{p}^{n}} \leq  (CAMERA_READY)
\left(\EE_{\eps_{[t]}}{\left\|\sum_{i=1}^{t}{\eps_i M_i}
\right\|_{\mathrm{C}_p^n}^p}\right)^{1/p} \leq B_p
\left\|\left(\sum_{i=1}^{t}{M_i^2}\right)^{1/2} \right\|_{\mathrm{C}_p^n}
\end{equation}
where for every $i\in{[t]}$, $\eps_i$ is a Bernoulli random variable. Moreover,
$B_p$ is at most\footnote{See Eqn.~(17) in~\cite{low_rank:STOC09} or~\cite{khintchine:Buchholz}.}
$2^{-1/4}\sqrt{\pi/e} \sqrt{p}$.
\end{theorem}
Now we are ready to prove Lemma~\ref{lem:E_p_vs_sum_of_squares}.
\begin{proof}(of Lemma~\ref{lem:E_p_vs_sum_of_squares})
The proof is inspired from~\cite[Theorem~3.1]{lowrank:rankone:VR}. Let $p\geq 2$.
First, apply a standard symmetrization argument (see~\cite{book:LedouxTalagrand}), which gives that 
\begin{equation*}
 \left(\EE_{M_{[t]}}{\left\|\frac1{t}\sum_{i=1}^{t}{M_i} -
\EE{M}\right\|^p_2}\right)^{\frac1{p}}\\
 \leq 2 \left(\EE_{M_{[t]}}\EE_{\eps_{[t]}}{{\norm{\frac1{t}\sum_{i=1}^{t}{\eps_i M_i}}^p}}\right)^{\frac{1}{p}}.
\end{equation*}
Indeed, let $\eps_1,\eps_2,\dots, \eps_t$ denote independent Bernoulli
variables. Let $M_1,\dots ,M_t,\widetilde{M}_1,\dots, \widetilde{M}_t$ be independent
copies of $M$. We essential estimate the $p$-th root of $E_p$,
\begin{equation}\label{eq:expect_moments}
 E_p^{1/p}\  =\ \left(  \EE_{M_{[t]}}{\norm{ \frac1{t} \sum_{i=1}^{t} M_i
-\EE M }^p} \right)^{1/p}
\end{equation}
Notice that $\EE{ \widetilde{M}}=\EE_{\widetilde{M_{[t]}}}{\left(\frac1{t}\sum_{i=1}^{t}
\widetilde{M_i}\right)}$. We plug this into \eqref{eq:expect_moments} and apply
Jensen's inequality,
\begin{eqnarray*}
 E_p^{1/p} & = & \left(  \EE_{M_{[t]}}{ \norm{\EE_{\widetilde{M}_{[t]}}{ \frac1{t}\sum_{i=1}^{t} M_i -\frac1{t}\sum_{i=1}^{t} \widetilde{M_i}} }^p} \right)^{1/p}\\
           & \leq & \left(  \EE_{M_{[t]} }{\EE_{\widetilde{M}_{[t]}}{ \norm{ \frac1{t}\sum_{i=1}^{t} M_i -\frac1{t}\sum_{i=1}^{t} \widetilde{M_i} }^p}} \right)^{1/p}.
\end{eqnarray*}
Now, notice that $M_i - \widetilde{M}_i$ is a symmetric matrix-valued random
variable for every $i\in{[t]}$, i.e., it is distributed identically with
$\eps_i(M_i -\widetilde{M}_i)$. Thus
\begin{equation*}
 E_p^{1/p} \leq \left( 
\EE_{M_{[t]}}{\EE_{\widetilde{M}_{[t]}}{\EE_{\eps_{[t]}}{\norm{\frac1{t}
\sum_{i=1}^{t} \eps_i(M_i - \widetilde{M_i})}^p}}} \right)^{1/p}.
\end{equation*}
Denote $Y=\frac1{t}\sum_{i=1}^{t} \eps_iM_i$ and
$\widetilde{Y}=\frac1{t}\sum_{i=1}^{t} \eps_i\widetilde{M_i}$.
Then $\|Y-\widetilde{Y}\|^p \leq (\|Y\| + \|\widetilde{Y}\|)^p \leq 2^p (\|Y\|^p
+\|\widetilde{Y}\|^p)$, and $\EE{\|Y\|^p}=\EE{\|\widetilde{Y}\|^p} $. Thus, we obtain
that
\begin{equation}\label{ineq:symmetrization}
 E_p^{1/p} \leq 2\left(\EE_{M_{[t]}}{\EE_{\eps_{[t]}}{ \left\|
\frac1{t}\sum_{i=1}^{t}\eps_i M_i\right\|^p_2}} \right)^{1/p}.
\end{equation}
Now by the Khintchine's inequality the following holds for any \emph{fixed}
symmetric matrices $M_1,M_2, \dots , M_t$.
\begin{eqnarray}
 	\left(\EE_{\eps_{[t]}} \norm{\frac1{t}\sum_{j=1}^{t}{\eps_j M_j} }^p \right)^{\frac1{p}}  
		      & \leq &  \frac1{t}\left(\EE_{\eps_{[t]}} \left\|\sum_{j=1}^{t}{\eps_j M_j} \right\|^p_{C_p} \right)^{\frac1{p}} \nonumber \\
		      & \leq &  \frac1{t} B_p \left\|\left(\sum_{j=1}^{t}{M_j^2}\right)^{1/2}\right\|_{C_p}\nonumber\\
		      & \leq &  \frac{(rt)^{1/p} B_p}{t} \norm{\left(\sum_{j=1}^{t}{M_j^2}\right)^{\frac1{2}}}\nonumber\\
		      &   =  & \frac{(rt)^{1/p} B_p}{t} \norm{\sum_{j=1}^{t}{M_j^2}}^{\frac1{2}} \label{ineq:final},
\end{eqnarray}
taking $1/t$ outside the expectation and using the left part of
Ineq.~\eqref{ineq:Schatter:rank}, Ineq.~\eqref{ineq:non_com:Khintchine}, the
right part of Ineq.~\eqref{ineq:Schatter:rank} and the fact that the matrix
$\left(\sum_{j=1}^{t}{M_j^2}\right)^{1/2}$ has rank at most $rt$.

Now raising Ineq.~\eqref{ineq:final} to the $p$-th power on both sides and
then take expectation with respect to $M_1,\dots ,M_t$, it follows from
Ineq.~\eqref{ineq:symmetrization} that 
\begin{equation*}
 		E_p\ \leq\ 2^p  \cdot \frac{rt}{t^p}B_p^p
\EE_{M_{[t]}}{\norm{\sum_{j=1}^{t}{M_j^2}}^{p/2}}.
\end{equation*}
This concludes the proof of Lemma~\ref{lem:E_p_vs_sum_of_squares}.
\end{proof}
%%%%%%%%%%%%%%%%%%%%%%%%%%%%%%%%%%%%%%%%%%%%%%%%%%%%%%%%%%%%
%%%%%%%%%%%%%%%%%%%%%%%%%%%%%%%%%%%%%%%%%%%%%%%%%%%%%%%%%%%%
%\paragraph{Part \textit{(b)}:}
%%%%%%%%%%%%%%%%%%%%%%%%%%%%%%%%%%%%%%%%%%%%%%%%%%%%%%%%%%%%
%%%%%%%%%%%%%%%%%%%%%%%%%%%%%%%%%%%%%%%%%%%%%%%%%%%%%%%%%%%%
Now we are ready to prove Theorem~\ref{thm:chernoff:matrix_valued:low_rank}. First we can assume without loss of generality that $M\succeq 0$ almost surely losing only a constant factor in our bounds. Indeed, by the spectral decomposition theorem any symmetric matrix can be written as $M=\sum_{j} \lambda_j u_j u_j^\top $. Set $M_{+} = \sum_{\lambda_j \geq 0} \lambda_j u_j u_j^\top$ and $M_{-}=M-M_{+}$. It is clear that $\norm{M_{+}},\norm{M_{-}} \leq \norm{M}$, $\frobnorm{M_{+}}, \frobnorm{M_{-}} \leq \frobnorm{M}$ and $\rank{M_{+}}, \rank{M_{-}} \leq \rank{M}$. Triangle inequality tells us that
\begin{eqnarray*}
\norm{\frac1{t} \sum_{i=1}^{t} M_j - \EE{M}} \ \leq\ \norm{\frac1{t} \sum_{i=1}^{t} (M_j)_{+} - \EE{M_{+}}}\\
+\ \norm{\frac1{t} \sum_{i=1}^{t} (M_j)_{-} - \EE{M_{-}}}  
\end{eqnarray*}
and one can bound each term of the right hand side separately. Hence, from now on we assume that $M \succeq 0$ a.s.. Now use the fact that for every $j\in{[t]}$, $M_j^2 \preceq \gamma \cdot
M_j$ since $M_j$'s are positive semi-definite and $\norm{M}\leq \gamma$
almost surely. Summing up all the inequalities we get that
\begin{equation}\label{ineq:psd:sum_of_squares}
\norm{\sum_{j=1}^{t}{M_j^2}} \leq \gamma \norm{\sum_{j=1}^{t}{M_j}}.
\end{equation}
It follows that
\begin{eqnarray*}
 	E_p	& \leq & r t^{1-p} (2B_p)^p	\EE_{M_{[t]}}{\norm{\sum_{j=1}^{t}{M_j^2}}^{p/2}}\\
		& \leq & r t^{1-p} (2B_p)^p\gamma^{p/2} \EE_{M_{[t]}}{\norm{\sum_{j=1}^{t}{M_j}}^{p/2}} \\
		&   =  &  \frac{rt(2B_p\sqrt{\gamma})^p}{t^{p/2}} \EE_{M_{[t]}}{\norm{\frac1{t}\sum_{j=1}^{t}{M_j}}^{p/2}} \\
		&   =  & \frac{rt(2B_p\sqrt{\gamma})^p}{t^{p/2}} \EE_{M_{[t]}}{\norm{\frac1{t}\sum_{j=1}^{t}{M_j} -\EE M + \EE M}^{p/2}} \\
%======================================================
% This is the old calculation of the paper.
%======================================================
%		& \leq & \frac{rt(2B_p\sqrt{\gamma})^p}{t^{p/2}} \left(\left(\EE_{M_{[t]}}{\norm{\frac1{t}\sum_{j=1}^{t}{M_j} - \EE M}^{p/2}}\right)^{2/p} + \norm{\EE M} \right)^{p/2}\\
%		& \leq & \frac{rt(2B_p\sqrt{\gamma})^p}{t^{p/2}} \left(\left(\EE_{M_{[t]}}{\norm{\frac1{t}\sum_{j=1}^{t}{M_j} -\EE M}^{p}}\right)^{1/p} + \norm{\EE M} \right)^{p/2} \\
%======================================================
%======================================================
		& \leq & \frac{rt(2B_p\sqrt{\gamma})^p}{t^{p/2}} \left(\left(\EE{\norm{\frac1{t}\sum_{j=1}^{t}{M_j} - \EE M}^{\frac{p}{2}}}\right)^{\frac{2}{p}} + 1 \right)^{\frac{p}{2}}\\
		& \leq & \frac{rt(2B_p\sqrt{\gamma})^p}{t^{p/2}} \left(\left(\EE{\norm{\frac1{t}\sum_{j=1}^{t}{M_j} -\EE M}^{p}}\right)^{\frac1{p}} + 1 \right)^{\frac{p}{2}} \\
%		&   =  & \frac{rt(2B_p\sqrt{\gamma})^p}{t^{p/2}}	 \left(E_p^{1/p} + \norm{\EE M} \right)^{p/2}\\
%======================================================
% This is the old calculation of the paper.
%======================================================
		&   =  & \frac{rt(2B_p\sqrt{\gamma})^p}{t^{p/2}} \left(E_p^{1/p} + 1 \right)^{p/2},
%======================================================
%======================================================
\end{eqnarray*}
using Lemma~\ref{lem:E_p_vs_sum_of_squares},
Ineq.~\eqref{ineq:psd:sum_of_squares}, Minkowski's inequality, Jensen's
inequality, definition of $E_p$ and the assumption $\norm{\EE{M}} \leq 1$. This
implies the following inequality
\begin{equation}
 E_p^{1/p} \ \leq \    \frac{2B_p\sqrt{\gamma} (rt)^{1/p}}{\sqrt{t}} (E_p^{1/p}
+ 1),
\end{equation}
using that $\sqrt{1+x}\leq 1+ x$, $x\geq 0$. Let $a_p=\frac{4B_p\sqrt{\gamma}
(rt)^{1/p}}{\sqrt{t}}$. Then it follows from the above inequality that $
E_p^{1/p} \leq \frac{a_p}{2} (E_p^{1/p} +1)$. It follows that\footnote{Indeed,
if $E_p^{1/p}<1$, then $E_p^{1/p} < a_p $. Otherwise $1\leq a_p$.} $\min
\{E_p^{1/p}, 1\} \leq a_p $. Also notice that 
\begin{equation}\label{ineq:exp_min}
\left(\EE \min \{Z, 1 \}^p\right)^{1/p} \leq \min ( E_p^{1/p}, 1).
\end{equation}
Now for any $0<\eps <1$,
\begin{equation*}
 \Prob{ Z >\eps }\ =\ \Prob{\min\{Z,1\} > \eps }.
\end{equation*}
By the moment method we have that 
\begin{eqnarray*}
 \Prob{\min\{Z,1\} > \eps } &=& \Prob{\min\{Z,1\}^p > \eps^p } \\
			    & \leq & \inf_{p\geq 2}\left( \frac{ \EE{\min\{Z,1\}^p} }{\eps^{p}}  \right)  \\
			    & \leq & \inf_{p\geq 2}\left( \frac{\min\{E_p^{1/p},1\}^p}{\eps^{p}} \right) \quad \eqref{ineq:exp_min}\\
			    & \leq & \inf_{p\geq 2} \left(\frac{a_p}{\eps}\right)^p \\
			    &   =  & \inf_{p\geq 2} \left(\frac{4B_p\sqrt{\gamma} (rt)^{1/p}}{\eps\sqrt{t}}\right)^p  \\
	    		    &   =  & \inf_{p\geq 2} \left(C_2\frac{\sqrt{p\gamma} (rt)^{1/p}}{\eps\sqrt{t}}\right)^p,
\end{eqnarray*}
where $C_2>0$ is an absolute constant.

Now assume that $r\leq t$ and then set $p=c_2\log t$, where
$c_2>0$ is a sufficient large constant, at the infimum expression in the above inequality, 
it follows that
\begin{eqnarray*}
\Prob{ \norm{ \frac1{t}\sum_{i=1}^{t}{M_i} - \EE M } >\eps }
&\leq&  \left(C\frac{\sqrt{\gamma\log t } (rt)^{\frac1{\log t}}}{\eps\sqrt{t}}\right)^{c_2\log t}
\end{eqnarray*}
We want to make the base of the above exponent smaller than one. It is easy to see that this is possible if we set 
$t=C_0 \gamma/\eps^{2} \log (C_0\gamma/\eps^{2})$ where $C_0$ is sufficiently large
absolute constant. Hence it implies that the above probability is at most $1/\poly{t}$. This concludes the proof.

\end{document}